\documentclass[aps,pra,twocolumn,10pt,superscriptaddress]{revtex4-1} 

\usepackage[english]{babel}
\usepackage[T1]{fontenc} %for what? That's why: http://tex.stackexchange.com/questions/664/why-should-i-use-usepackaget1fontenc
\usepackage{times}

\usepackage{amsthm} %for theorem environments
\usepackage{paralist} %for compactenum
\usepackage{amssymb} %lots of math symbols
\usepackage{mathtools} %for, e.g., \DeclareMathOperator
\usepackage{amsmath} %for spacing in equations and aligns
\usepackage{mathbbol} %for \1, load at last
\usepackage{color}
\usepackage{graphicx}
\usepackage[pdfpagelabels]{hyperref} %will be done by the arxiv 
\hypersetup{colorlinks=true} 

\newtheorem{theorem}{Theorem}
\newtheorem{lemma}{Lemma}
\newtheorem{corollary}{Corollary}
\newtheorem{observation}{Observation}
\newtheorem{definition}{Definition}
\newtheorem{implication}{Implication}

\DeclareMathOperator{\Tr}{Tr} %Use either \DeclareMathOperator or \operatorname http://tex.stackexchange.com/questions/84302/what-is-the-difference-of-mathop-operatorname-and-declaremathoperator
\DeclareMathOperator{\supp}{supp}
\DeclareMathOperator{\cov}{cov}
\DeclareMathOperator{\dist}{d}

\newcommand{\e}{\mathrm{e}}
\newcommand{\rmd}{\mathrm{d}}
\DeclareMathOperator{\landauO}{O}

\newcommand{\1}{\mathbb{1}}
\newcommand{\NN}{\mathbb{N}}

\newcommand{\CC}{\mathbb{C}}
\newcommand{\ZZ}{\mathbb{Z}}

\newcommand{\mc}[1]{\mathcal{#1}}
\renewcommand{\H}{\mc{H}}
\DeclareMathOperator{\G}{\mc{G}}
  
\newcommand{\norm}[1]{\left\Vert #1\right\Vert}
\newcommand{\nnorm}[1]{\Vert #1 \Vert}
\newcommand{\bnorm}[1]{\bigl\| #1 \bigr\|}
\newcommand{\Bnorm}[1]{\Bigl\| #1 \Bigr\|}

\newcommand{\ad}{^\dagger}
\newcommand{\kw}[1]{\frac{1}{#1}}
\newcommand{\restr}{\upharpoonright}
\renewcommand{\ol}[1]{\overline{#1}}
\DeclareMathOperator{\dunion}{\uplus}%Disjoint union
\DeclareMathOperator{\dUnion}{\biguplus}%Disjoint union
\renewcommand{\complement}{^{c}}
\DeclareMathOperator{\colonequiv}{:\!\Leftrightarrow}
\newcommand{\ket}[1]{\left|#1 \right\rangle}

\newcommand{\trunc}[2]{#1_{\upharpoonright \mathnormal{#2}}} %truncation
\newcommand{\gibbs}{g}
\newcommand{\animalc}{\alpha}
\DeclareMathOperator{\C}{\mc{C}}
\DeclareMathOperator{\A}{\mc{A}}
\newcommand{\clusters}[3][F]{\C_{#2}^{#3}\!(#1)}
\newcommand{\animals}[3][F]{\A_{#2}^{#3}(#1)}
\newcommand{\W}{\mc{W}} %for different sets of words
\DeclareMathOperator{\boundary}{\partial}
\DeclareMathOperator{\alphay}{\mathnormal{b}}
\newcommand{\swap}{\mc{S}}
\newcommand{\Eset}{E}
\newcommand{\Fset}{F}
\newcommand{\I}{{(1)}}
\newcommand{\II}{{(2)}} 

\newcommand{\titletext}{Locality of temperature}

\newcommand{\fu}{Dahlem Center for Complex Quantum Systems, Freie Universit{\"a}t Berlin, 14195 Berlin, Germany}
\newcommand{\aei}{Max Planck Institute for Gravitational Physics (Albert Einstein Institute), 
	     Am M\"uhlenberg 1, 14476 Potsdam-Golm, Germany}

\newcommand{\listkeys}{Gibbs state, canonical state, canonical temperature, clustering of correlations, critical temperature, spin lattice system, spin lattice model, classical simulation, intensive temperature, perturbation of thermal states, stability of thermal states, thermalization, lattice system, local Hamiltonian, thermodynamic limit, matrix product operators, MPO, locality of temperature, intensive temperature, intensivity, intensiveness, quantum information, quantum thermodynamics, lattice animal, cluster expansion, cluster analysis
}

\newcommand{\listpacs}{
  05.30.-d,  % Quantum statistical mechanics
  05.50.+q, % Lattice theory and statistics (Ising, Potts, etc.)
  03.67.-a,  % Quantum information
}

\makeatletter %pdf meta data (with hyperref)
 \hypersetup{pdftitle={\titletext},
	     pdfauthor={M. Kliesch, C. Gogolin, M.J. Kastoryano, A. Riera, J. Eisert},
	     pdfsubject={Subject Areas: Condensed Matter Physics, Quantum Physics, Statistical Physics; PACS numbers: \listpacs},
	     pdfkeywords={\listkeys}
	    }
\makeatother
     
\begin{document}
% Set pdf meta data (without hyperref)
% \pdfinfo{ %doesn't work with hyperref
% 	/Author (M. Kliesch, C. Gogolin, M.J. Kastoryano, A. Riera, J. Eisert) 
% 	/Title (\titletext) 
% 	/Subject (Subject Areas: Quantum Physics, Statistical Physics; PACS numbers: \listpacs)
% 	/Keywords (\listkeys)
% 	} 

\title{\titletext}

\author{M.\ Kliesch} \affiliation{\fu}
\author{C.\ Gogolin} \affiliation{\fu}
\author{M.~J.\ Kastoryano} \affiliation{\fu}
\author{A.\ Riera} \affiliation{\fu} \affiliation{\aei}
\author{J.\ Eisert} \affiliation{\fu}

% \pacs{\listpacs}

% \keywords{\listkeys}

\begin{abstract}
This work is concerned with thermal quantum states of Hamiltonians on spin and fermionic lattice systems with short range interactions. We provide results leading to a local definition of temperature, thereby extending the notion of ``intensivity of temperature'' to interacting quantum models. More precisely, we derive a perturbation formula for thermal states. The influence of the perturbation is exactly given in terms of a generalized covariance. For this covariance, we prove exponential clustering of correlations above a universal critical temperature that upper bounds physical critical temperatures such as the Curie temperature. As a corollary, we obtain that above the critical temperature, thermal states are stable against distant Hamiltonian perturbations. Moreover, our results imply that above the critical temperature, local expectation values can be approximated efficiently in the error and the system size.
\end{abstract}

\maketitle
%%%===================================================================
%%%======================      Intro      ============================
%%%===================================================================
\section{Introduction}
\label{sec:intro}
The ongoing miniaturization of devices, with structures reaching the nanoscale, has lead to the development of extremely small thermometers \cite{PotGueBir97,PengPeng13}, some of which are so small that they can only be read out with powerful electron microscopes \cite{GaoBan02}. Even small thermal machines  working in the quantum regime have been suggested \cite{LinPopSkr10,MarEis11}.
In order to understand the working of such devices, it is necessary to formulate a theory of statistical mechanics and thermodynamics at the microscopic and mesoscopic scales. 
A prerequisite for this is a good understanding of the limitations of the concept of temperature at small scales. 

The problem with assigning locally a temperature to a small subsystem of a globally thermal system is the following: 
Interactions between the subsystem and its environment that generate correlations can lead to noticeable deviations of the state of the subsystem from a thermal state (see Figure~\ref{fig:intensivetemperatureproblem}). 
Hence, given only a subsystem state, there is no canonical way to assign a temperature to the subsystem. 
We call this the \emph{locality of temperature problem}. 

The first steps toward a solution of the locality of temperature problem have been taken in  Refs.~\cite{HarMahHess04-PRL,Har06,HarMah05}, and more recently, within the mindset of quantum information theory, in Ref.~\cite{Ferraro_intensive_T}. 
The general locality of temperature problem is, however, still open.
In this work, we conclusively solve it for spin and fermionic lattice systems. 

More precisely, we first show that the locality of temperature problem is equivalent to a decay of correlations measured by an averaged generalized covariance that precisely captures the response of expectation values to perturbations of the Hamiltonian. 
We expect the corresponding equality to be useful for applications beyond the scope of this article.

We then provide conditions under which the generalized covariance decays exponentially with the distance, including a detailed analysis of the preasymptotic, and of the finite-size regime. 
In particular, this exponential decay holds above a universal critical temperature that only depends on the ``connectivity'' of the underlying graph of the model and is an upper bound on physically relevant critical temperatures such as the Curie temperature.

While, in the low-temperature regime, quantum lattice models exhibit a great diversity of phases, many of which involve the emergence of long-range or topological order \cite{Vojta}, in the high-temperature regime, exponential clustering of correlations is expected. 
Our rigorous results help to delineate the boundary between these two regimes.
They build upon and go significantly beyond previous results on the clustering of correlations in classical systems \cite{Rue99_stat_mech_book}, for quantum gases \cite{Gin65}, i.e., translation-invariant Hamiltonians in the continuum, and cubic lattices \cite{BraRob97_op_alg_book,Gre69, Park-Yoo}.

\begin{figure}[t]
\centering
\includegraphics{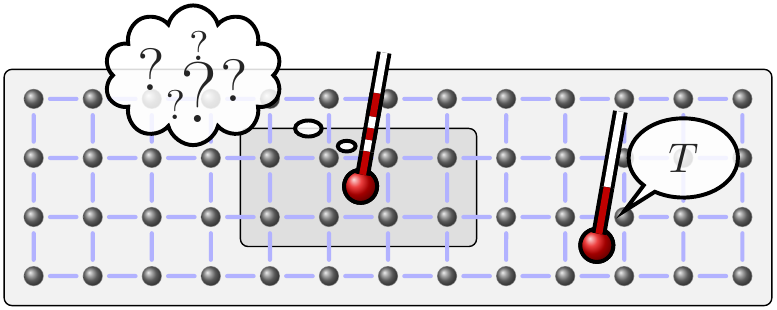}
\caption{The locality of temperature problem: Subsystems of thermal states are themselves, in general, not in a state with a locally well-defined temperature. 
Down to which length scale can temperature be an intensive quantity?}
\label{fig:intensivetemperatureproblem}
\end{figure}

Mathematically, we significantly contribute to the problem of whether and under which precise conditions thermal quantum states are stable against distant Hamiltonian perturbations. 
This is particularly relevant in the broader scheme of phase transitions in classical and quantum lattice models \cite{BhaKha1995,BraRob97_op_alg_book} 
as well as for the foundations of statistical mechanics and the equilibration and thermalization behavior of closed quantum systems \cite{GolLebTum06,CraDawEis08,RigDunOls08,LinPopSho09,Yuk2011,ReiKas12,Ours,MueAdlMas13,MasRonAugAci13}.
In the light of the recent surge of interest in these topics, developing a better understanding of the properties of thermal states has become a timely issue.

A major obstacle to progress on some of the most interesting open questions in this context, such as equilibration time scales in closed quantum systems, is the limited set of mathematical tools available for exploiting the structure of locally interacting Hamiltonians \cite{MasRonAugAci13}. 
Our results are among the first that explicitly exploit properties of local Hamiltonians, without being limited to very specific models. 

For quantum Monte Carlo simulations \cite{TroAleTre03}, our results provide a guideline as to how 
large the finite system size has to be taken in order to be able to sample from the right partition function and, conversely, to identify observables that are best suited to detect long-range correlations.

In fact, our results are reminiscent of known statements about ground states.
If a Hamiltonian has a unique ground state and is gapped, correlations in its ground state cluster exponentially and faraway regions become essentially uncorrelated. 
This is rigorously proven using information theory inspired methods such as Lieb-Robinson bounds and quasi-adiabatic continuation \cite{Koma,NacSim06,NacOgaSim06}. 
These rigorous results allow for certified algorithms that efficiently approximate ground states of gapped Hamiltonians on classical computers \cite{LanVazVid13}.
In the same spirit, we are able to show that an exponential decay of correlations renders thermal states locally efficiently simulatable. 

The rest of this paper is structured as follows: 
In Section~\ref{sec:settingandmainresults}, we formulate the precise setting and explain the main results and their implications. 
In Section~\ref{sec:connections}, we discuss connections to known results on phase transitions, thermalization in closed quantum systems, and matrix product operator approximations. 
Then, in Section~\ref{sec:details}, we discuss basic properties of the generalized covariance, explain how our results can be made applicable to finite-range $k$-body interactions, and state the results for fermionic lattices. 
We proceed with proving all theorems in Section~\ref{sec:proofs} and conclude in Section~\ref{sec:conclusions}. 
In the Appendix, we provide a detailed proof of two bounds on truncated cluster expansions, one of which is an important ingredient to the proof of clustering of correlations.

%%%===================================================================
%%%======================      Setting      ==========================
%%%===================================================================
\section{Setting and main results}
\label{sec:settingandmainresults}
In this section, we introduce the setting, state the locality of temperature problem more formally, and state our results.

\subsection{Perturbation formula for thermal states}
As the first result, we state a perturbation formula, which is a general statement about the response of the expectation value of an observable in the thermal state, upon changes in the system Hamiltonian.
It does not make any reference to the locality structure of the Hamiltonian but turns out to be especially useful when correlations between local observables decay rapidly with distance. 

Throughout the paper, we assume the Hilbert space to be finite dimensional \cite{Note1} and denote the \emph{thermal state}, or Gibbs state, of a Hamiltonian $H$ at \emph{inverse temperature} $\beta$ by 
\begin{equation}
  \gibbs(\beta) \coloneqq \frac{\e^{-\beta\,H}}{Z(\beta)} \, ,
\end{equation}
with $Z(\beta) \coloneqq \Tr(\e^{-\beta\,H})$ being the \emph{partition function}. 
If we mean the thermal state or partition function of a different Hamiltonian $H'$, we write $\gibbs[H'](\beta)$ or $Z[H'](\beta)$.

We measure correlations by the \emph{(generalized) covariance} that we define for any two operators $A$ and $A'$, full-rank quantum state $\rho$, and parameter $\tau \in [0,1]$ as
\begin{equation}\label{eq:def_cov}
 \cov_\rho^\tau(A,A') 
 \coloneqq \Tr\left(\rho^\tau A\, \rho^{1-\tau} A'\right) - \Tr(\rho\, A) \Tr(\rho \, A') \, .
\end{equation}
We discuss various properties of this covariance and generalizations to arbitrary-rank quantum states in Section~\ref{sec:generalizedcovariance}. 

The generalized covariance appears naturally in our first theorem about the response of expectation values to perturbations. More precisely, when we are given an unperturbed Hamiltonian $H_0$ and perturbed Hamiltonian $H$, then the difference of expectation values in the corresponding thermal states is captured by that covariance: 

\begin{theorem}[Perturbation formula]\label{thm:perturbation_formula}
Let $H_0$ and $H$ be Hamiltonians acting on the same Hilbert space. 
For $s \in [0,1]$, define the \emph{interpolating Hamiltonian} by 
$H(s) \coloneqq H_0 +s\,(H-H_0)$ and denote its thermal state by
$\gibbs_s \coloneqq \gibbs[H(s)]$. 
Then, 
\begin{equation}\label{eq:perturbation_formula}
  \begin{split}
    \Tr\bigl(A\,\gibbs_0(\beta)\bigr) 
	& - \Tr\bigl(A\,\gibbs(\beta)\bigr)\\
    &= \beta\int_0^1 \rmd \tau \int_0^1 \rmd s\, \cov_{\gibbs_s(\beta)}^\tau (H-H_0,A) 
  \end{split}
\end{equation}
for any operator $A$.
\end{theorem}
The proof of the theorem, which is presented in Section~\ref{sec:proof_perturbation_formula}, relies on the fundamental theorem of calculus and Duhamel's formula.
We refer to the double integral over the covariance in Eq.~\eqref{eq:perturbation_formula} as the \emph{averaged (generalized) covariance}.

\subsection{Spin lattice systems}
In the remainder of this work, we will be concerned with spin and fermionic lattice systems.
We will only write out everything for spin systems and then later, in Sections~\ref{sec:fermionic_results} and~\ref{sec:fermionic_proofs}, explain the necessary modifications for fermionic systems. 
In the case of spin lattice systems, the Hilbert space is given by 
$\H = \bigotimes_{x \in V} \H_x$, where $V$ is called the \emph{vertex set} and is assumed to be finite.
To make the presentation more accessible, many of the following definitions are highlighted in Figure~\ref{fig:lattice}.
A \emph{local Hamiltonian} with \emph{interaction (hyper)graph} $(V,\Eset)$ is 
a sum
\begin{equation}\label{eq:def_local_H}
  H = \sum_{\lambda \in \Eset} h_\lambda 
\end{equation}
of \emph{local Hamiltonian terms} $h_\lambda$ acting on $\H$.  
The \emph{(hyper)edge set} $\Eset$ is the set of supports $\lambda = \supp(h_\lambda) \subset V$ of the local terms $h_\lambda$. 
For any subset of edges $\Fset \subset \Eset$, we denote by 
$H_\Fset \coloneqq \sum_{\lambda \in \Fset} h_\lambda$ the Hamiltonian only containing the interactions in $\Fset$, and for any subsystem $B \subset V$, we define the \emph{truncated Hamiltonian} to be
$\trunc H B \coloneqq H_{\Eset(B)}$, where 
$\Eset(B) \subset \{\lambda \in \Eset: \lambda \subset B\}$ is the \emph{restricted edge set} and we take $\trunc H B$ to be an operator on the Hilbert space $\H_B \coloneqq \bigotimes_{x \in B} \H_x$.

Given some subsystem $S\subset V$ there are two natural thermal states associated with it:
\begin{compactenum}[(i)]
  \item 
  $\trunc \gibbs S(\beta) \coloneqq \gibbs[\trunc H S](\beta)$ denotes the thermal state of $S$ alone, i.e., the thermal state of the truncated Hamiltonian $\trunc{H}{S}$.
  \item 
  $\gibbs^S(\beta) \coloneqq \Tr_{S\complement}( \gibbs(\beta))$ denotes the full thermal state \emph{reduced} to $S$.
\end{compactenum}
For a non-interacting Hamiltonian, these two states coincide, but, in general, this is not the case due to correlations between $S$ and its environment.
This discrepancy raises the question of how to define temperature locally as an intensive quantity, i.e., the \emph{locality of temperature problem}.

\begin{figure}[t]
\centering
\includegraphics{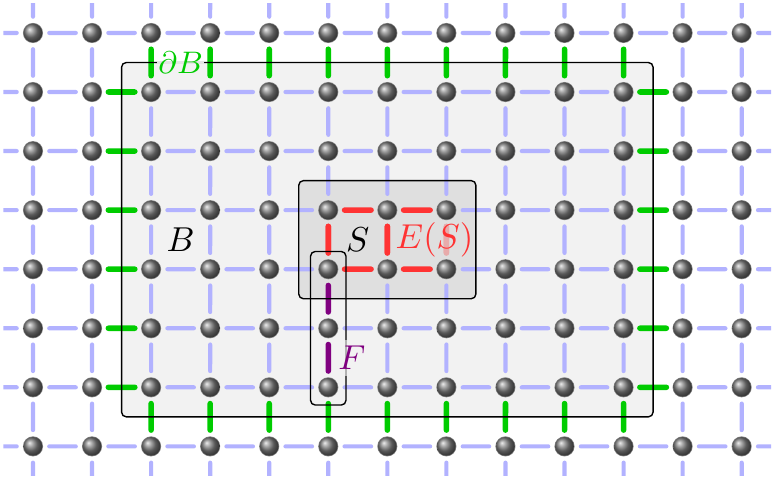}
\caption{
  A 2D square lattice: The boxes indicate subsystems $S\subset B \subset V$. 
  The edges in $S$ are $\Eset(S)$, boundary edges of $B$ are 
  $\boundary\! B$, 
  and~$\Fset$ is a shortest path connecting $S$ and $\boundary\! B$; hence, 
  $\dist(S,\boundary\! B)=|\Fset|=2$. 
  The set of edges $\Eset(S)$ is an example for an animal of size $|\Eset(S)|=7$, while $\boundary\! B$ is not connected and hence not an animal.
  }
\label{fig:lattice}
\end{figure}

\subsection{Locality of temperature}
In order to  locally assign a temperature to the subsystem $S \subset V$ it was suggested, e.g., in Ref.~\cite{Ferraro_intensive_T}, to extend $S$ by a buffer region and define the temperature of $S$ via the thermal state of the Hamiltonian truncated outside the extended region $B$, see Figure~\ref{fig:lattice}. 
The role of the buffer region $B$ is to remove the boundary effects and the correlations with the rest of the system that are intuitively the reason for the locality of temperature problem.
Nevertheless, it is not obvious how these correlations should be quantified and how large this buffer region needs to be. 
We will see shortly that Theorem~\ref{thm:perturbation_formula} answers these questions.

By $\boundary\! B \subset \Eset$, we denote the set of \emph{boundary edges} of $B$, i.e., the edges having overlap with both $B$ and its complement $B\complement \coloneqq V \setminus B$.
Then, by choosing $H_0 = H - H_{\boundary\! B}$ in Theorem~\ref{thm:perturbation_formula}, using that $\gibbs_0 = \trunc{\gibbs}{B} \otimes \trunc{\gibbs}{B\complement}$, and tracing over $B\complement$, we obtain the following:

\begin{corollary}[Truncation formula]\label{cor:truncation_formula}
Let $H$ be a local Hamiltonian, let $B \subset V$ be a subsystem, and denote the corresponding \emph{boundary Hamiltonian} by $H_{\boundary\! B}$, and the \emph{interpolating Hamiltonian} by 
$H(s) \coloneqq H - (1-s)\, H_{\boundary\! B}$ with its thermal state $g_s\coloneqq \gibbs[H(s)]$.
Then, for any operator $A=A_B \otimes \1_{B\complement}$ supported on $B$, 
\begin{equation}\label{eq:truncation_error_in_terms_of_cov}
  \begin{split}
    \Tr\bigl[A_B\,\trunc{\gibbs}{B}(\beta)\bigr] 
	& - \Tr\bigl[A\,\gibbs(\beta)\bigr]\\
    &= \beta\int_0^1 \rmd \tau \int_0^1 \rmd s\, \cov_{\gibbs_s(\beta)}^\tau (H_{\boundary\! B}, A) .    
  \end{split}
\end{equation}
\end{corollary}

Now we choose $S \subset B \subset V$ (see Figure~\ref{fig:lattice}). 
If, for a given inverse temperature $\beta$, correlations over the distance between $S$ and $\boundary\! B$ are negligible, then the corollary clearly implies that 
\begin{equation}
\Tr[A\, \gibbs(\beta)] \approx \Tr[A_B\,\trunc{\gibbs}{B}(\beta)]
\end{equation}
for any observable 
$A_B = A_S \otimes \1_{B\setminus S}$ on $S$.
Also note that such an approximate equality does not hold whenever average correlations over lengths exceeding the distance between $S$ and $\boundary\! B$ are non-negligible.

Hence, we have the following equivalence for the temperature defined via thermal states:

\begin{implication}[Locality of temperature] 
 Temperature is intensive on a given length scale if and only if correlations (measured by the averaged generalized covariance) are negligible compared to $1/\beta$ on that length scale. 
\end{implication}

In order to fully exploit Corollary~\ref{cor:truncation_formula}, it is necessary to bound the generalized covariance, which we will do for high temperatures in the next section. 

\subsection{Clustering of correlations at high temperatures}
For small temperatures, correlations can be arbitrarily long-ranged, as is, e.g., the case for the ferromagnetic Ising model in two or higher dimensions below the Curie temperature.
On the other hand, above a universal critical temperature, depending only on a local property of the interaction graph, correlations cluster exponentially, as we will see next. 
Given the combinatorial nature of parts of the arguments leading to this result, 
we need additional notation related to edges and vertices of the lattice. Most of the following definitions can be understood intuitively, as is shown in Figure~\ref{fig:lattice}.

We say that two subsystems $X,Y \subset V$ \emph{overlap} if $X \cap Y \neq \emptyset$, 
a set $X \subset V$ and a set $\Fset\subset \Eset$ \emph{overlap} if $\Fset$ contains an edge that overlaps with $X$, and two sets $\Fset,\Fset'\subset \Eset$ \emph{overlap} if $\Fset$ overlaps with any of the edges in $\Fset'$. 
A subset of edges $\Fset \subset \Eset$ \emph{connects} $X$ and $Y$ if $\Fset$ contains a sequence of pairwise overlapping edges such that the first overlaps with $X$ and the last overlaps with $Y$ and similarly for the case where $X$ and/or $Y$ are just vertices.

The graph distance on $V$, and also the induced distance on subsets of $V$, are denoted by $\dist$. 
The distance $\dist(X,\Fset)$ of a subset $X\subset V$ and a subset $\Fset \subset \Eset$ is $0$ if $X$ and $\Fset$ overlap and otherwise equal to the size of the smallest subset of $\Eset$ that connects $X$ and $\Fset$. 
Sometimes, we denote the support of an operator by the operator itself, e.g., for two operators $A$ and $A'$, their distance is 
$\dist(A,A') \coloneqq \dist(\supp A, \supp A')$ and 
$\boundary\! A \subset \Eset$ are the edges across the boundary of $\supp(A)$.

A subset of edges $\Fset \subset \Eset$ that connects all pairs of its elements $\lambda,\lambda' \in \Fset$ is called \emph{connected}.
Such a connected set $\Fset$ is also called an \emph{(edge) animal}. 
The size $|\Fset|$ of an animal $\Fset$ is given by the number of edges contained in $\Fset$.
The results presented here apply to Hamiltonians with interaction graphs $(V,\Eset)$ whose number $a_m$ of lattice animals of size $m$ containing some fixed edge is exponentially bounded. 
With
\begin{equation}\label{eq:animal_bound}
  a_m \coloneqq \sup_{\lambda \in \Eset} |\{ \Fset\subset \Eset \ \text{ connected}: \lambda \in \Fset,\ |\Fset|=m \}| \, ,
\end{equation}
the \emph{growth constant} $\animalc$ is the smallest constant satisfying
\begin{equation} \label{eq:animal_const}
  a_m \leq \animalc^m .
\end{equation}
For example, the growth constant of a $D$-dimensional cubic lattice can be bounded as $\animalc \leq 2\,D\,\e$ (Lemma~2 in Ref.~\cite{MirSlade2010}), where $\e $ is Euler's number. 
Moreover, $\animalc$ is finite for any regular lattice \cite{Pen94}. 
Upper bounds to growth constants for so-called spread-out graphs \cite{MirSlade2010} render our results applicable for the case of bounded-range two-body interactions. 
By a simple embedding argument, one can also bound the growth constant for the case of local $k$-body interactions on a regular lattice, which we explain in Section~\ref{sec:embedding} in detail. 

For any operator $A$ and $p \in [1,\infty]$, we denote by $\norm{A}_p$ its Schatten $p$-norm; e.g., $\norm{A}_\infty$ is the operator norm and $\norm{A}_1$ is trace norm of $A$. 
We call $J\coloneqq \max_{\lambda \in \Eset} \norm{h_\lambda}_\infty$ the \emph{local interaction strength} of a local Hamiltonian, as given in Eq.~\eqref{eq:def_local_H}.

We are able to provide a universal inverse \emph{critical temperature} $\beta^\ast$, which is, in particular, independent of the system size, below which correlations decay exponentially with a \emph{thermal correlation length} $\xi(\beta)$:
  
%%% --------------------------------------------------------------
%%% -------------      Clustering Theorem      -------------------
%%% --------------------------------------------------------------
\begin{theorem}[Clustering of correlations at high temperatures] \label{thm:clustering}
Let $\gibbs(\beta)$ be the thermal state at inverse temperature $\beta$ of a local Hamiltonian with finite interaction (hyper)graph $(V,\Eset)$ having growth constant $\animalc$ and local interaction strength $J$. 
Define the quantities
\begin{equation}\label{eq:crit_beta_def}
\beta^\ast \coloneqq \ln\bigl[\bigl(1+\sqrt{1+4/\animalc}\bigr)/2\bigr]/(2\,J) 
\end{equation}
and
\begin{equation} \label{eq:correlation_length_def}
\xi(\beta) \coloneqq \bigl|\ln\bigl[\animalc\, \e^{2\,|\beta|\,J}\bigl(\e^{2\,|\beta|\,J}-1\bigr)\bigr]\bigr|^{-1}  \, .
\end{equation}
Then, for every $|\beta|<\beta^\ast$, parameter $\tau \in [0,1]$, every two operators $A$ and $B$ with 
$\dist(A, B) \geq L_0(\beta,a)$ [given in Eq.~\eqref{eq:L_zero_def}], 
and $a \coloneqq \min\{|\boundary\! A |,|\boundary\! B|\}$, 
\begin{equation} \label{eq:clustering} 
 |\cov^\tau_{\gibbs(\beta)} (A, B)| 
 \leq 
 \frac{4 \, a\, \norm{A}_\infty\, \norm{B}_\infty}  {\ln(3)\, (1-\e^{-1/\xi(\beta)})} \, 
 \e^{-\dist(A,B)/\xi(\beta)} .
\end{equation}
\end{theorem}

The proof is given in Section~\ref{sec:proof_clustering}.

In the following sections, we outline some of the applications of Theorem~\ref{thm:clustering}.

\subsection{Universal locality and stability at high temperatures}
\label{sec:stability_locality}
If one is interested in the state $\gibbs^S(\beta)$ of some subsystem $S$, then one can truncate the Hamiltonian to $S$ extended by some buffer region and obtain the approximation via the thermal state of the truncated Hamiltonian.
The following theorem implies that the approximation error is exponentially small in the width of the buffer region.

For any operator $\rho$, we denote its \emph{reduction} to a subsystem $S \subset V$ by
$\rho^S \coloneqq \Tr_{S\complement}[\rho]$ and note that 
\begin{equation} \label{eq:one_norm_imterpretatoin}
\bnorm{\rho^S}_1 = \sup\left\{ |\Tr[A\,\rho]|: \supp(A) = S,\, \norm{A}_\infty =1\right\}.
\end{equation}
Then, as a consequence of Corollary~\ref{cor:truncation_formula} and Theorem~\ref{thm:clustering}, we obtain the following:
  
\begin{figure}
\centering
\includegraphics{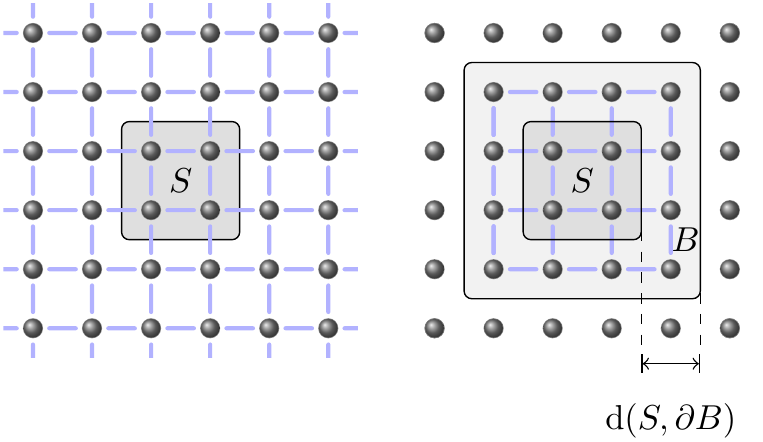}
\caption{
  The truncation from Corollary~\ref{cor:intensivity} and~\ref{imp:efficient_approx}: For $\beta< \beta^\ast$ and $\dist(S,\boundary\! B) \ll \xi(\beta)$ Corollary~\ref{cor:intensivity} implies that $\gibbs^S(\beta)$, depicted on the left, and $\trunc{\gibbs}{B}^S(\beta)$, depicted on the right, are approximately equal.}
\label{fig:buffer-non-buffer}
\end{figure}

%%% ----------------------------------------------------------------------------
%%% ------------------------ Intensive temperature thm. ------------------------
%%% ----------------------------------------------------------------------------
\begin{corollary}[Universal locality at high temperatures]\label{cor:intensivity}
Let $H$ be a Hamiltonian satisfying the conditions of Theorem~\ref{thm:clustering}, 
let $|\beta|< \beta^\ast$, and let $S \subset B \subset V$ be subsystems with 
$\dist(S, \boundary\! B) \geq L_0(\beta, |\boundary\! S|)$. 
Then, 
\begin{equation}
 \norm{\gibbs^S(\beta)-\trunc{\gibbs}{B}^{S}(\beta)}_1  
 \leq
 \frac{ v\, |\beta|\, J } {1-\e^{-1/\xi(\beta)}} \, 
 \e^{- \dist(S, \boundary\! B) /\xi(\beta)} ,
\end{equation}
where $\trunc{\gibbs}{B}^S$ denotes the thermal state of $B$ reduced to $S$ and $v \coloneqq 4\, |\boundary\! S|\,|\boundary\! B|/\ln(3)$. 
\end{corollary}

Similarly, as a corollary of Theorems~\ref{thm:perturbation_formula} and~\ref{thm:clustering} we obtain the following:

\begin{implication}[Stability]\label{impl:stability}
  Below the critical inverse temperature $\beta^*$ [from Eq.~\eqref{eq:crit_beta_def}], thermal states of local Hamiltonians are exponentially stable against distant locally bounded perturbations.
\end{implication}

\subsection{Efficient approximation}
Corollary~\ref{cor:intensivity} on the universal locality of thermal states also has the following  complexity theoretic consequence:

\begin{implication}[Efficient approximation]\label{imp:efficient_approx}
  For $|\beta| < \beta^\ast$, local expectation values can be approximated with a computational cost independent of the system size and bounded polynomially in the reciprocal error. 
\end{implication}

In this sense, the error bound (see Figure~\ref{fig:cone}) of Corollary~\ref{cor:intensivity} is reminiscent of the \emph{quasi-locality of dynamics}, as, e.g., presented in Ref.~\cite{BarKli12}, which is a consequence of \emph{Lieb-Robinson bounds} \cite{Pou10,NacVerZag11}. 
The quasi-locality theorem \cite{BarKli12} allows for an approximation of time evolved local observables by truncating the Hamiltonian in the time evolution operator at a distance $L>0$ far away from the space time cone of the observable's support and has an approximation error that is exponentially small in $L$.

\begin{figure}
\centering
\includegraphics[width=.95\linewidth]{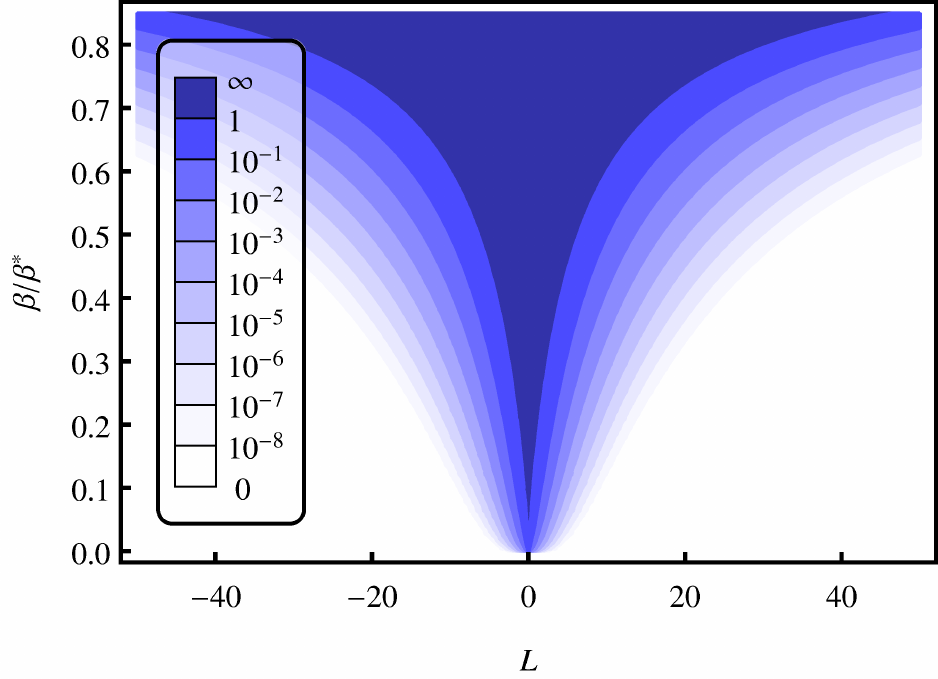}
\caption{One can obtain slightly tighter error bounds in Corollaries~\ref{cor:intensivity} and~\ref{cor:intensivity_fermions} by using Eq.~\eqref{eq:cov_intermediate_bound_general} directly. 
The plot shows this bound on the approximation error $\norm{\gibbs^S(\beta)-\trunc{\gibbs}{B}^{S}(\beta)}_1$ for the case of $S$ being a single site on a 2D square lattice as a function of the inverse temperature $\beta$ in units of the critical temperature and the width of the buffer region $L$. 
This can be seen as an imaginary time Lieb-Robinson ``cone'' with diverging width as $\beta \to \beta^\ast$.
}
\label{fig:cone}
\end{figure}

\subsection{Fermions}
In Ref.~\cite{Has04_fermions}, it was shown for \emph{fermionic systems} that two-point functions of observables that are odd polynomials in the fermionic operators decay exponentially with a correlation length proportional to the inverse temperature. 
Here, we obtain an exponential decay of the covariance above the critical temperature for all operators.

\begin{observation}[Fermions]\label{obs:fermions}
All results also hold for locally interacting fermions on a lattice.
See Theorem~\ref{thm:clustering_fermions} and Corollaries~\ref{cor:truncation_formula_fermions} and~\ref{cor:intensivity_fermions} in Section~\ref{sec:fermionic_results} for the precise statements. 
\end{observation}

%%%=========================================================
%%%=====================  Connections ======================
%%%=========================================================
\section{Relations to known results}
\label{sec:connections}
In this section, we discuss the critical temperature from the clustering theorem, 
the connection of this work to concepts related to thermalization, and
approximations of thermal states with so-called \emph{matrix product operators}. 
As a last point, we briefly mention similarities with local topological quantum order. 

\subsection{Critical temperatures and phase transitions}
Our results show that the  quantity $\beta^\ast$, as defined in Eq.~\eqref{eq:crit_beta_def}, provides a potentially coarse but universal and completely general upper bound on physical critical temperatures like the \emph{Curie temperature}.
For the ferromagnetic two-dimensional isotropic Ising model without external field, our bound yields, for example. $1/(\beta^\ast\,J) = 2/\ln((1+\sqrt{1+1/\e})/2) \approx 24.58$, 
whereas the phase transition between the disordered paramagnetic and the ordered ferromagnetic phases is known to really happen at
$1/(\beta_c\,J) = 2/\ln(1+\sqrt{2}) \approx 2.27$
\cite{BhaKha1995}. 
Our universal bound is about an order of magnitude higher than the actual value for this example. 
To put this discrepancy into perspective, it is worth pointing out that it is generally a very difficult task to estimate physical critical temperatures --- numerically or analytically. 
In fact, analytic expressions for critical temperatures or even just bounds on their values are known only for very few models. 

One of the few known general statements is the \emph{Mermin-Wagner-Hohenberg theorem} \cite{MerWag1966}. 
It states that in certain low-dimensional systems with short-range interactions there cannot be any phase transition involving the spontaneous breaking of a continuous symmetry at any non-zero temperature. 
However, such systems can still have a low-temperature phase with quasi-long-range order characterized by power-law-like decaying correlations. 
Consequently, even for systems covered by the Mermin-Wagner-Hohenberg theorem, our Theorem~\ref{cor:intensivity} is nontrivial. 
For example, it implies an upper bound on the critical temperature of the Kosterlitz-Thouless transition in the two-dimensional $XY$-model \cite{FroSpe1981}. 

In this work, we have concentrated on the general picture, but it seems likely that refinements of the methods employed and developed here can yield much tighter bounds on critical temperatures if more specific properties of a model are taken into account. 
At the same time, it remains an open problem to actually find a model with a phase transition with long-range order at the universal highest possible temperature. 

\subsection{Foundations of statistical mechanics}
The recent years have seen a large number of numerical and experimental (see Ref.~\cite{Yuk2011} for a review) as well as analytical investigations (see, for example, Refs.~\cite{GolLebTum06,CraDawEis08,LinPopSho09,ReiKas12,Ours,MueAdlMas13}) of equilibration and thermalization in closed quantum systems.
In the focus of these works are the approach to equilibrium or properties of energy eigenstates.
The current work complements this body of literature in that it shows fundamental properties of systems in thermal equilibrium.
A feature that makes the current work unique is that, contrary to essentially all other works, the results derived here explicitly use the structure of locality interacting systems (noteworthy exceptions are Ref.~\cite{MueAdlMas13} and, albeit in a very special setting, Ref.~\cite{CraDawEis08}).

The locality of thermal states is also of interest for recent results \cite{MueAdlMas13} on the dynamical thermalization of translation-invariant lattice models: Our Corollary~\ref{cor:intensivity} guarantees the existence of a ``unique phase'' \cite{MueAdlMas13} for all temperatures above our critical temperature.
Hence, it implies that at sufficiently high temperatures, Theorems~1, 2, and 3 of Ref.~\cite{MueAdlMas13} are applicable for any translation-invariant Hamiltonian.

There is also an interesting connection of our locality and stability results to the so-called \emph{eigenstate thermalization hypothesis} (ETH) \cite{RigDunOls08,Yuk2011}.
The ETH essentially conjectures that the expectation values of certain physically relevant observables (for example local ones) in energy eigenstates of sufficiently complex Hamiltonians should be very similar to the expectation values in thermal states with the same average energy.
Corollary~\ref{cor:intensivity} and Implication~\ref{impl:stability} thus imply that the eigenstates of a Hamiltonian in the center of the spectrum (which correspond to high-temperature thermal states) must, if the Hamiltonian fulfills the ETH, also be locally stable against perturbations of the Hamiltonian.
This insight could put constraints on the class of Hamiltonians that fulfills the ETH, provide new insights into the properties of their eigenstates, and open up new ways to test the ETH.

\subsection{MPO approximation of thermal states}\label{sec:hastings_MPO}
\emph{Matrix Product Operators} (MPOs) are a certain class of operators that are tractable on classical computers for one-dimensional systems. Therefore, they play an important role in numerical simulations based on so-called \emph{tensor networks}. 

An important ingredient to our proof of Theorem~\ref{thm:clustering} on clustering of correlations will be a bound on a truncated cluster expansion (Lemma~\ref{lem:tailored_hastings}). 
The original result on the cluster expansion (Lemma~\ref{lem:hastings} in the Appendix) is due to  Hastings and was first used to approximate thermal states with inverse temperature below $2\,\beta^\ast$ by MPOs \cite{Hastings06}. 
This approximation is summarized in the next theorem.

In one spatial dimension, this MPO approximation yields a tensor size bounded polynomially in the system size and the approximation error (see the subsequent corollary).
In higher dimensions, however, the MPO approximation yields a tensor size bounded only subexponentially in the system size and is hence not computationally efficient, albeit exponentially cheaper than storing the full density matrix $\gibbs(\beta)$. 
In order to explain this in more detail, we start the discussion with a slightly non-standard definition of MPOs:
\begin{definition}[Matrix product operator (MPO)] $\ $ \\ 
Let $(b[x]^{(j)})_{j=1}^{d^2}$ be a basis for the operators on $\H_x$ and write an arbitrary operator $A$ on $\H$ in the product basis as
 \begin{equation}
  A = \sum_{k \in [d^2]^{V}} A_k \bigotimes_{x \in V} b[x]^{(k_{x})} ,
 \end{equation}
 with expansion coefficients $A_k \in \CC$ and where $[d^2] \coloneqq \{1,2,\dots,d^2\}$. 
 If the $A_k$ are of the form
 \begin{equation} \label{eq:MPO_def}
  A_k = \prod_{x \in V} a[x](k) \, ,
 \end{equation}
 where every $a[x](k)$ only depends on at most $r$ of the $|V|$ indices $k_x$, then $A$ is called an \emph{MPO with tensor size $d^{2r}$}. 
\end{definition}

Thermal states can be approximated by such MPOs.
The following theorem is a consequence of Lemma~\ref{lem:hastings}, which we will prove in the Appendix along with Lemma~\ref{lem:tailored_hastings}. 

\begin{theorem}[MPO approximation of thermal states \cite{Hastings06}] \label{thm:MPO_approx}
Let $H=\sum_{\lambda \in \Eset} h_\lambda$ be a local Hamiltonian with finite interaction graph $(V,\Eset)$ having a growth constant $\animalc$, and local interaction strength  
$J = \max_{\lambda \in \Eset} \norm{h_\lambda}_\infty$, 
and define 
$\alphay(\beta J) \coloneqq \animalc \, \e^{|\beta J|}\,\left(\e^{|\beta J|} -1\right)$.
Moreover, let $\beta$ be small enough such that $\alphay(\beta J)<1$. 
Then, for each $L \in \ZZ^+$, there exists a self-adjoint MPO $\rho(\beta,L)$ [given in Eq.~\eqref{eq:MPO_approx}] with tensor size $d^{2\, N(L)}$, where 
\begin{equation}
 N(L)\coloneqq \sup_{x_0 \in V} |\{ x \in V : \dist(x,x_0)<L \}|
\end{equation}
is the number of vertices within a distance less than $L$.
The approximation error is bounded as
\begin{equation}\label{eq:MPO_approx_error}
 \norm{\gibbs(\beta) - \rho(\beta,L)}_1 
 \leq \exp\! \left(|E|\frac{\alphay(\beta J)^L}{1-\alphay(\beta J)} \right) - 1 \, ;
\end{equation}
i.e., for fixed $|\beta J|<\alphay^{-1}(1)$, the trace norm difference scales as 
$\landauO(|E|\, \e^{-|\ln[\alphay(\beta J)]|\,L})$ for large enough $L$.
\end{theorem}

In particular, the theorem implies the following:
\begin{corollary}[Bound on the tensor size]\label{cor:tensor_size}
Let $D$ be the spatial dimension of the Hamiltonian's interaction graph $(V,E)$, let $n\coloneqq |E|$ be the system size, and $\beta<2\,\beta^\ast$ with $\beta^\ast$ from Eq.~\eqref{eq:crit_beta_def}. 
Then, the MPO approximation in Theorem~\ref{thm:MPO_approx} gives rise to a tensor size of the MPO $\rho(\beta,L)$ scaling as
\begin{equation}
  \log_d(\mathrm{tensor\ size}) \leq \landauO(\ln(C\,n/\epsilon)^D) \, ,
\end{equation}
with some $\beta$-dependent constant $C$. 
In particular, for $D=1$, the bound on the tensor size scales polynomially with $n/\epsilon$.
\end{corollary}

Let us consider a one-dimensional system and suppose we are explicitly given the MPO tensors $a'[x]$ [see Eq.~\eqref{eq:MPO_def}] of an approximation to a state $\rho$ and, similarly, an observable $A$ of MPO form with MPO tensors $a[x]$. 
If the tensor sizes of both MPOs scale at most polynomially in the system size, then one can compute the corresponding approximation to the expectation value $\Tr(\rho A)$ with a computational cost scaling polynomially in the system size. 
This means that, for instance, \emph{global} product observables can be approximated efficiently, which is not guaranteed by our Implication~\ref{imp:efficient_approx}. The problem with the MPO approximation, however, is that Theorem~\ref{thm:MPO_approx} only guarantees the existence of the MPO tensors but it is not obvious how they can be computed (efficiently). 

\begin{proof}[Proof of Corollary~\ref{cor:tensor_size}]
The condition $\beta<2\,\beta^\ast$ is equivalent to $\alphay(\beta J)<1$. 
Let us denote the bound to the approximation error in Eq.~\eqref{eq:MPO_approx_error} by $\epsilon$.
Note that the upper bound in Eq.~\eqref{eq:MPO_approx_error} satisfies
\begin{equation}
 \epsilon \coloneqq 
 \exp\! \left(|E|\frac{\alphay(\beta J)^L}{1-\alphay(\beta J)} \right) - 1 
 \leq C \, n \, \alphay(\beta J)^L 
\end{equation}
for distances $L$ being at least logarithmically large in $n = |E|$ and some $\beta$-dependent constant $C$. 
Then, the distance $L$ necessary to reach $\epsilon$ must asymptotically be at least as large as
\begin{equation} 
L \geq \frac{\ln( C \, n/\epsilon)}{|\ln[\alphay(\beta J)]|} \, .
\end{equation}
Bounding $N(L)$ in terms of the spatial dimension $D$ as 
$N(L) \leq M\,L^D$ with some constant $M$ yields a tensor size bounded as
\begin{align}
 \log_d(\mathrm{tensor\ size}) 
 \leq 
 2\,M\,\left(\frac{\ln\left(C \, n/\epsilon\right)}{\ln\left(1/\alphay(\beta J)\right)}\right)^D.
\end{align}%
\end{proof}

\subsection{Local topological quantum order}
It is worth mentioning that Corollary~\ref{cor:intensivity} and Implication~\ref{impl:stability} are very reminiscent of the \emph{local topological quantum order} condition for open quantum systems introduced in Ref.~\cite{Cubitt} and the results on the local stability of stationary states of local Liouvillians in Ref.~\cite{MJK}.
A slightly different family of local topological quantum order conditions for closed quantum systems \cite{Cubitt, MJK, Michalakis1,Michalakis2} has played a very important role in the theory of locally stable (topological) lattice systems and for rigorous proofs of entropic area laws. Corollary~\ref{cor:intensivity} similarly characterizes the regime where local perturbations cannot drive any thermal phase transition.

%%%=========================================================
%%%=================  Details ======================
%%%=========================================================
\section{Details}\label{sec:details}
In this section, we first discuss the generalized covariance and then provide details concerning the applicability of our results to Hamiltonians with $k$-body interactions. 
Finally, we justify Observation~\ref{obs:fermions} by stating the fermionic versions of our results.

%%%=========================================================
%%%================  Generalized Covariance =====================
%%%=========================================================
\subsection{The generalized covariance}
\label{sec:generalizedcovariance}
The generalized covariance defined in Eq.~\eqref{eq:def_cov}, which depends on a parameter $\tau \in [0,1]$, provides more information about the correlations between two observables than the standard covariance in a similar way as the class of R\'enyi entropies characterizes more completely the entanglement properties of a state than simply the von Neumann entropy \cite{EntSpec}.
While it occurs quite naturally in the perturbation formula Theorem~\ref{thm:perturbation_formula}, 
other possible applications are to be explored. 
Here, we discuss possible generalizations of the generalized covariance to operators of arbitrary rank, 
show that for operators $A$ and $A'$ they are always bounded by $\norm{A}_\infty \norm{A'}_\infty$, 
and comment on convexity and a symmetrized version of the generalized covariance. 

A definition of the generalized covariance for states of arbitrary rank is not relevant for this work because for non-zero temperatures thermal states are full-rank operators. 
However, the discussion of possible generalizations also hints at the behaviour of $\cov^\tau$ at the end points of the unit interval. 
On the open interval $\tau \in {]0,1[}$, it is natural to simply keep the definition from Eq.~\eqref{eq:def_cov}. 
There are two natural ways to define $\rho^0$:
Either as, $\rho^0 \coloneqq \1$ or as $\rho^{0+} \coloneqq \lim_{\tau \to 0} \rho^\tau$, where $\rho^{0+}$ turns out to be the projector onto the image of the operator $\rho$. 
For each end point $\tau = 0$ and $\tau = 1$, there are hence two natural ways to define $\cov^\tau$, either such that the generalized covariance is continuous or such that $\cov^0_\rho(A,A')=\cov_\rho(A',A)$ and $\cov^1_\rho(A,A')=\cov_\rho(A,A')$, where 
\begin{equation}
\cov_\rho(A,A') \coloneqq \Tr(\rho \, A\, A') - \Tr(\rho\, A) \, \Tr(\rho \, A') \,   
\end{equation}
defines the standard covariance. 

Note that for product states and operators with disjoint support, all versions of the generalized covariance vanish. 
Moreover, for pure states, the continuous version of the generalized covariance vanishes also, meaning that classical correlations are needed to yield a non-zero value. 

Next, we show that the generalized covariance is always bounded as
\begin{equation}\label{eq:cov_bound} 
 \bigl|\cov^\tau_\rho(A,A') \bigr| \leq \norm{A}_\infty \norm{A'}_\infty \, , 
\end{equation}
irrespective of which definitions are chosen for $\cov^0$ and $\cov^1$. 
We consider a state $\rho$ and define $\bar{A}\coloneqq A-\Tr(\rho A)$. 
Then, 
\begin{equation}
  \cov_\rho^\tau(A,A') = \Tr\left(\rho^\tau \bar{A}\, \rho^{1-\tau} A'\right) .
\end{equation}
H\"older's inequality generalized to several operators and the fact that 
$\norm{X}_p = \norm{|X|^p}_1^{1/p}$ then imply that 
\begin{align}
  \bigl|\cov_\rho^\tau(A,A')\bigr| 
  &\leq 
  \norm{\rho^\tau}_{1/\tau} \, \nnorm{\bar{A}}_\infty \, \norm{\rho^\tau}_{1/(1-\tau)} \, \norm{A'}_\infty \\
  &= \nnorm{\bar{A}}_\infty \, \norm{A'}_\infty 
\end{align}
and, by noting that $\nnorm{\bar A}_\infty = \norm{A}_\infty$, the bound~\eqref{eq:cov_bound} is proven for the continuous version of the generalized covariance. 
For the non-continuous versions, the bound follows similarly. 

The variance $\cov_\rho^\tau(A,A)$ induced by the continuous version of the covariance is convex in $\tau$. This can be seen by writing out $\rho$ in its eigenbasis. 
As one can change the sign of $\cov_\rho^\tau(A,A')$ by just changing the sign of $A'$, 
the generalized covariance is not convex in $\tau$. 
But, it might be that its magnitude $|\cov_\rho^\tau(A,A')|$ is convex, which is unclear. 
If this were the case, it would be enough to prove the clustering Theorem~\ref{thm:clustering} only for the end points $\tau \in \{0,1\}$, and hence the proof could be significantly simplified. 

Similarly, as there is a symmetrized version of the standard covariance, one can also symmetrize the generalized covariance with respect to the two operators. 
Because of the cyclicity of the trace, the generalized covariance satisfies the symmetry property
\begin{equation} \label{eq:generalized_covariance_symmetry property}
  \cov_\rho^\tau(A,A') = \cov_\rho^{1-\tau}(A',A) 
\end{equation}
Hence, one can define the symmetrized version of the generalized covariance as follows:
  \begin{equation}
    \ol{\cov}_\rho^\tau(A,A') \coloneqq \frac{1}{2} \left( \cov_\rho^\tau(A,A') + \cov_\rho^\tau(A',A) \right) .
  \end{equation}
Our results can also be phrased in terms of this symmetrized version, since the averaged generalized covariance in the perturbation formula Theorem~\ref{thm:perturbation_formula} can easily be rewritten in terms of $\ol{\cov}$, and a bound analogous to the clustering Theorem~\ref{thm:clustering} holds also for the symmetrized quantity.

%%%=========================================================
%%%=================  Lattice animals ======================
%%%=========================================================
% \subsection{Bound on the growth constant for local $k$-body interactions}
\subsection{Bound on the growth constant for local \texorpdfstring{$k$-}\ body interactions}
\label{sec:embedding}
In this section, we show that regular hyperlattices also have a finite growth constant, which renders our results applicable to Hamiltonians with local $k$-body interactions. 

In the case of $k$-body interactions, the Hamiltonian is again a sum of local terms $h_{\lambda}$ whose supports are hyperedges 
$\lambda = \supp(h_{\lambda}) \subset V$ with $|\lambda|\leq k$. 
As before, $V$ denotes the vertex set and $E$ the set of hyperedges. 

We assume that the interaction hypergraph $(V,E)$ is a \emph{regular hyperlattice}, i.e., that it can be embedded into a regular hypercubic lattice of a certain dimension $D$ with hyperedges of hypercubic form.
Let us denote by $R$ the edge length of the resulting hypercubes.
Note that such an embedding is, in general, not unique and changes both the number of terms in the Hamiltonian and the local interaction strength of $H$. 
Moreover, the grouping changes the values of the metric $\dist$ in our results. 

In order to find an exponential upper bound to the number $a_m$ of hyperanimals composed of $m$ hypercubes, let us define a \emph{spread-out} graph of range $R$  as the graph with the edge set consisting of all pairs $\{x, y\}$ with $0 < \norm{x - y}_\infty \leq R$ and $x,y\in \ZZ^D$ (see Ref.~\cite{MirSlade2010}). 
Notice that as any hypercube is uniquely specified by the coordinates of its ``lower left corner'', any hyperanimal of size $m$ corresponds to a lattice animal of size $m-1$ and range $R$ in the spread-out graph. 
It follows from Lemma~2 in Ref.~\cite{MirSlade2010} that $a_m\leq (K \e)^m$ with $K = (2\,R+1)^D-1$ being the coordination number. 
Hence, the hyperlattice has a growth constant bounded by $\animalc \leq ((2\,R+1)^D-1)\, \e$. 

The bound obtained is for most models, far from optimal, in particular, in situations where the supports of the local Hamiltonian terms are very different from hypercubes. 
For such cases, one can derive tighter but more specific bounds from known results about lattice animals in a similar way.

\subsection{Fermionic versions of the main results}
\label{sec:fermionic_results}
To make Observation~\ref{obs:fermions} about fermions precise, we introduce the setting of interacting fermions on lattices.
For each site $x \in V$, the corresponding fermionic operators, i.e., the creation and annihilation operators $f_x\ad$ and $f_x$, act on the fermionic Fock space and satisfy
\begin{align}
  \{ f_x ,f_y\ad \} &= \delta_{x,y}\, \1 \, , \\
  \{ f_x ,f_y \} &= 0 \, ,
\end{align}
where $\{A,B\}\coloneqq A\,B + B\,A$ is the \emph{anti-commutator}. 
For such systems, all operators can be given in terms of polynomials in the fermionic operators.
A monomial of fermionic operators is called \emph{even (odd)} if it can be written as a product of an even (odd) number of fermionic operators $f_x$ and $f\ad_y$.
A polynomial of fermionic operators is called \emph{even (odd)} if it can be written as a linear combination of only even (odd) monomials, and an operator is called \emph{even (odd)} if it can be written as an even (odd) polynomial of fermionic operators.
According to the \emph{fermion number parity superselection rule}, only operators that are even polynomials in the fermionic operators are physical observables and Hamiltonians.

As with spin lattice systems, we have again a finite interaction graph $(V,E)$; however, the support of an operator is now to be understood in the picture of second quantization as follows: 
The \emph{support} of any operator $A$ being a polynomial in the fermionic operators is the set of vertices of the fermionic operators that occur in the polynomial.
Correspondingly, we denote the algebra of the even operators supported on a region $X \subset V$ by $\G_X$ and denote $\G\coloneqq \G_V$ for short.
The Hamiltonian of a fermionic lattice system is of the form
\begin{equation}
  H = \sum_{\lambda \in E} h_\lambda 
\end{equation}
with $h_\lambda \in \G_\lambda$.
For $B \subset V$, the truncated Hamiltonian $\trunc H B$ is similarly the sum only over the edges contained in $B$. 
As for spin systems, $H_{\boundary\! B}$ is the sum over the boundary edges of $B$.

Theorem~\ref{thm:perturbation_formula} also holds for such fermionic lattice systems, and we can prove statements analogous to Corollary~\ref{cor:truncation_formula}, Theorem~\ref{thm:clustering}, and Corollary~\ref{cor:intensivity}. 
Hence, all implications stated in Section~\ref{sec:settingandmainresults}, also hold. 
All proofs are presented in Section~\ref{sec:fermionic_proofs}.

\begin{corollary}[Fermionic truncation formula]\label{cor:truncation_formula_fermions}
 Let $H= \sum_{\lambda \in E} h_\lambda$ be a fermionic local Hamiltonian with local terms $h_\lambda \in \G$, 
 let $B \subset V$ be a subsystem, and let the \emph{interpolating Hamiltonian} by $H(s) \coloneqq H - (1-s)\, H_{\boundary\! B}$ with its thermal state $g_s\coloneqq \gibbs[H(s)]$.
Then, for any operator $A$ with support $\supp(A) \subset B$,
\begin{equation}\label{eq:truncation_error_in_terms_of_cov_fermions}
  \begin{split}
    \Tr\bigl(A\,\gibbs[\trunc H B](\beta)\bigr) 
	& - \Tr\bigl(A\,\gibbs(\beta)\bigr)\\
    &= \beta\int_0^1 \rmd \tau \int_0^1 \rmd s\, \cov_{\gibbs_s(\beta)}^\tau (H_{\boundary\! B}, A) .    
  \end{split}
\end{equation}
\end{corollary}

\begin{theorem}[Clustering of correlations in fermionic systems]\label{thm:clustering_fermions}
Let $\gibbs(\beta)$ be the thermal state at inverse temperature $\beta$ of a local fermionic Hamiltonian 
$H= \sum_{\lambda \in E} h_\lambda$ 
with finite interaction graph $(V,E)$ having growth constant $\animalc$, local terms $h_\lambda \in \G$, and local interaction strength $J$. Define the functions $\beta^\ast$, $\xi$, and $L_0$ as in Eqs.~\eqref{eq:correlation_length_def}, \eqref{eq:crit_beta_def}, and \eqref{eq:L_zero_def}. 
Then, for every $|\beta|<\beta^\ast$, $\tau \in [0,1]$, and every two operators $A$ and $B$ with 
$\dist(A, B) \geq L_0(\beta,a)$,
where $a \coloneqq \min\{|\boundary\! A |,|\boundary\! B|\}$, 
\begin{equation}
 |\!\cov_{\gibbs(\beta)}^\tau (A, B)| 
 \leq 
 \frac{4 a \norm{A}_\infty\, \norm{B}_\infty}  {\ln(3)\, (1-\e^{-1/\xi(\beta)})} \, 
 \e^{-\dist(A,B)/\xi(\beta)} .
\end{equation}
\end{theorem}

\begin{corollary}[Locality of fermionic thermal states]\label{cor:intensivity_fermions}
Let $H$ be a Hamiltonian satisfying the conditions of Theorem~\ref{thm:clustering_fermions}, 
let $|\beta|< \beta^\ast$, and let $S \subset B \subset V$ be subsystems with 
$\dist(S, \boundary\!B) \geq L_0(\beta, |\boundary\! S|)$.
Then, 
\begin{equation} 
\norm{\gibbs^S(\beta)  - \gibbs^S[\trunc H B](\beta)}_1
 \leq
 \frac{ v\, |\beta|\, J } {1-\e^{-1/\xi(\beta)}} \, 
 \e^{- \dist(S, \,\boundary\! B) /\xi(\beta)} ,
\end{equation}
where $v = 4\, |\boundary\! S|\,|\boundary\! B|/\ln(3)$.
\end{corollary}

\section{Proofs}
\label{sec:proofs}
We start this section with the proofs of Theorems~\ref{thm:perturbation_formula} and~\ref{thm:clustering}. 
One important stepping stone for the proof of the latter is a tailored version of a bound on a truncated cluster expansion (Lemma~\ref{lem:tailored_hastings}) from Ref.~\cite{Hastings06}. 
Both versions are proven in the Appendix. 
In the last part of the section, we prove the fermionic versions of our main results, 
Therorem~\ref{thm:clustering_fermions} and Corollaries~\ref{cor:truncation_formula_fermions} and~\ref{cor:intensivity_fermions}.

\subsection{Proof of the perturbation formula (Theorem~\ref{thm:perturbation_formula})}
\label{sec:proof_perturbation_formula} 
The two main ingredients in the proof of Theorem~\ref{thm:perturbation_formula} are the fundamental theorem of calculus, and Duhamel's formula. 
The generalized covariance appears as a natural measure of correlations. 

\begin{proof}[Proof of Theorem~\ref{thm:perturbation_formula}]
Using the fundamental theorem of calculus we obtain
\begin{equation*}% use equation* in order to avoid a warning from hyperref
    \Tr[A\,\gibbs_0(\beta)] - \Tr[A\,\gibbs_1(\beta)]
  = -\Tr\!\left(A \int_0^1 \frac{\rmd}{\rmd s} \frac{\e^{-\beta\, H(s)}} {Z_s(\beta)} \, \rmd s \right) 
\end{equation*}
with $Z_s\coloneqq Z[H(s)]$.
The derivative can be written as
\begin{equation*}
 \frac{\rmd}{\rmd s} \frac{\e^{-\beta\, H(s)}} {Z_s(\beta)} = \frac{1}{Z_s(\beta)} \frac{\rmd}{\rmd s}\, \e^{-\beta H(s)} - \frac{\gibbs_s(\beta)}{Z_s(\beta)} \Tr\!\left(\frac{\rmd}{\rmd s}\, \e^{-\beta H(s)} \right) . 
\end{equation*}
After applying Duhamel's formula to both derivatives, i.e., using that
\begin{equation}\nonumber
\begin{split}
&\frac{\rmd}{\rmd s}\, \e^{ -\beta H(s)}\\
=&
-\beta \int_0^1 \left(\e^{-\beta H(s)}\right)^{\tau} \left(\frac{\rmd}{\rmd s}H(s)\right)\, \left(\e^{-\beta H(s)}\right)^{1-\tau}\, \rmd \tau\, ,
\end{split}
\end{equation}
we obtain 
\begin{multline}
\phantom{={}}{} \Tr[A\,\gibbs_0] - \Tr[A\,\gibbs]
\\
\shoveleft{=- \beta \Tr\biggl( A \int_0^1\int_0^1 \Bigl(- \gibbs_s^\tau
\,(H-H_0)\, \gibbs_s^{1-\tau}}
\\
\shoveright{\hfill+ \gibbs_s \Tr[ \gibbs_s^\tau \,(H-H_0)\,
\gibbs_s^{1-\tau}] \Bigr)\rmd \tau\, \rmd s \biggr) .} 
\end{multline}
Together with the the cyclicity of the trace and the definition of the generalized covariance in Eq.~\eqref{eq:def_cov}, this finishes the proof. 
\end{proof}

%%%=========================================================
%%%======  Hastings Lemma and Pf. of clustering ============
%%%=========================================================
\subsection{Proof of Theorem~\ref{thm:clustering} on clustering of correlations}
\label{sec:proof_clustering}
The proof of Theorem~\ref{thm:clustering} builds on and develops further a cluster expansion of the power series of $\e^{-\beta\,H}$ in terms of summands of the form 
\begin{equation}
  h(w) \coloneqq h_{w_1}\,h_{w_2} \dots h_{w_{|w|}} \, ,
\end{equation}
where $w_j \in \Eset$.
For the sake of a compact presentation, we refer to edges as \emph{letters}, to the edge set $E$ as an \emph{alphabet}, and call sequences of edges \emph{words}. 
For any \emph{sub-alphabet} $\Fset\subset \Eset$, we denote by 
$\Fset^\ast \coloneqq \bigcup_{l=0}^\infty \Fset^l$  
the set of words with letters in $\Fset$ and arbitrary length $l$, where the length $|w|$ of a word $w \in \Eset^\ast$ is the total number of letters it contains.
For two words $w,v \in \Eset^\ast$, their concatenation is denoted by 
$w\circ v\coloneqq (w_1,w_2,\dots,w_{|w|},v_1,v_2,\dots,v_{|v|})$. 
We call a word $c \in \Eset^\ast$ \emph{connected} or a \emph{cluster} if the set of letters in $c$ is an animal, i.e., connected.
So, clusters are connected sequences of edges where the edges can also occur multiple times, while animals are sets of edges without any order or repetition. 
A word $v$ is called a \emph{sub-sequence} of $w \in \Eset^\ast$ if $v$ can be obtained from $w$ by omitting letters, i.e., if there is an increasing sequence $j_1<j_2< \ldots<j_{|v|}$ such that $v_i = w_{j_i}$. 
This will be denoted by $v \subset w$. 
A connected sub-sequence $c \subset w$ is called a \emph{maximal cluster} of $w$ if $c$ is not a sub-sequence of any other connected sub-sequence of $w$.
Importantly, for any word $w \in \Eset^\ast$, one can permute its letters to a new word $w'$ such that 
$h(w') = h(w)$ 
irrespective of the choice of the local terms $h_\lambda$ and such that 
$w'=c_1\circ c_2\circ\dots \circ c_k$ is a concatenation of maximal clusters $c_j \subset w$ of $w$.
Note that this decomposition is unique up to the order of the $c_j$. 

In the following, we will consider systems that are either $n=2$ or $n=4$ copies of the original system with Hilbert space $\H$. 
For any operator $A$ on $\H$, we denote by $A^{(j)}$ the operator on $\H^{\otimes n}$ that acts as $A$ on the $j$th copy, e.g., $A^\II \coloneqq \1 \otimes A$ for $n=2$. 
By $\swap^{(i,j)}$, we denote the swap operator on $\H^{\otimes n}$ that swaps the $i$th and $j$th tensor factors, e.g., 
$\swap^{1,2} \ket{k_1,k_2,k_3,k_4} = \ket{k_2,k_1,k_3,k_4}$ for $n=4$. 
For $n=2$, we write $\swap$ instead of $\swap^{1,2}$. 

%%% ----------------------------------------------------------------------------
%%% ---------------------------- tailored Hastings Lemma --------------------------------
%%% ----------------------------------------------------------------------------
We can now state the subsequent lemma, which is a bound on a truncated cluster expansion that is based on a more general, but for our purposes not tight enough bound, used previously in Ref.~\cite{Hastings06} (see Lemma~\ref{lem:hastings} in the Appendix).
The lemma will play an important role in the subsequent proof of Theorem~\ref{thm:clustering}.

\begin{lemma}[Truncated cluster expansion]\label{lem:tailored_hastings}
Let $\tau \in [0,1]$ and $H=\sum_{\lambda \in \Eset} h_\lambda$ be a local Hamiltonian on $\H$ with finite interaction graph $(V,\Eset)$ having growth constant $\animalc$ and local interaction strength $J = \max_{\lambda \in \Eset} \norm{h_\lambda}_\infty$. 
We denote by $\tilde H$ the Hamiltonian of two weighted copies with local terms  $\tilde h_\lambda \coloneqq \tau \, h_\lambda^\I + (1-\tau) \, h_\lambda^\II$. 
Consider two operators $A$ and $B$ on $\H$, define 
$\alphay(x) \coloneqq \animalc \, \e^{|x|}\,\left(\e^{|x|} -1\right)$, and 
let $|\beta|$ be small enough such that $\alphay(\beta J)<1$.
For some set of edges $\Fset \subset \Eset$, let $\clusters{\geq L}{}\subset \Eset^\ast$ be the set of words containing at least one cluster $c$ that contains at least one letter of $\Fset$ and has size $|c|\geq L$ and let us denote the corresponding truncated cluster expansion of $\e^{-\beta \tilde H}$ by
\begin{equation}\label{eq:tildeOmega_def}
 \Omega[\tilde H](\beta) \coloneqq \sum_{w\in \clusters{\geq L}{}} \frac{(-\beta)^{|w|}}{|w|!} \tilde h(w) \, ,
\end{equation}
with $\tilde h(w) \coloneqq \tilde h_{w_1} \, \tilde h_{w_2} \dots \tilde h_{w_{|w|}}$. 
Then, for all $\tau \in [0,1]$, 
\begin{equation} \label{eq:tailored_Hastings_bound}
 \frac{\bigl|\Tr\bigl[\swap\, A^\I \, B^\II \, \Omega[\tilde H](\beta)\bigr]\bigr|}
      {\norm{A}_\infty \norm{B}_\infty Z(\beta)}
 \leq \exp\left({|\Fset|\frac{\alphay(\beta J)^L}{1-\alphay(\beta J)} }\right) - 1 \, .
\end{equation}
\end{lemma}

\begin{figure}[t]
\centering
\includegraphics{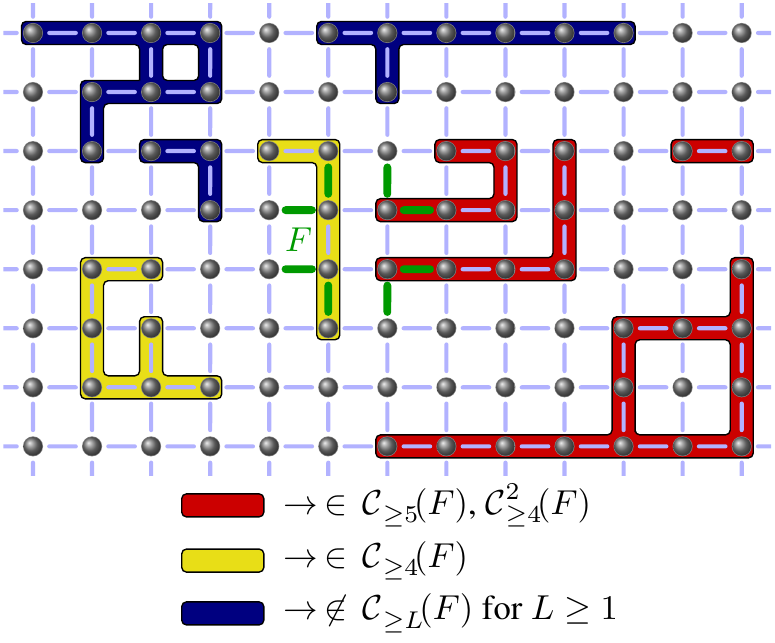}
\caption{A 2D square lattice. 
Three different sub-alphabets are indicated: 
Words that contain all letters in those alphabets are members of different sets $\clusters{\geq L}{}$. }
\label{fig:clusters}
\end{figure}

In the Appendix, we provide a detailed proof of this lemma. 
The terms resulting from the expansion of the exponential series of $\e^{-\beta H}$ are classified according to whether they contain a cluster of size at least $L$ that contains a letter from $\Fset$. 
One can then show that there is a percolation transition at $\beta^\ast = \alphay^{-1}(1)/(2J)$ such that for $|\beta|<\beta^\ast$, the contribution of long clusters is exponentially suppressed.

In the following proof of the exponential clustering, we will use the so-called \emph{swap-trick}: For any two operators $A$ and $B$, it holds that 
\begin{equation}\label{eq:swap-trick}
 \Tr(A\, B) = \Tr(\swap (A\otimes B)) \, ,
\end{equation}
which can be checked by a straightforward calculation.
  
\begin{proof}[Proof of Theorem~\ref{thm:clustering}]
Fix some $\tau \in [0,1]$.
For any operator $A:\H\to\H$, we define 
$A^{(\pm)} \coloneqq A\otimes \1 \pm \1\otimes A$ and 
$\tilde A^{(+)} \coloneqq \tau\,\bigl( A^\I + \,A^\II\bigr) + (1-\tau)\bigl( A^{(3)} + A^{(4)}\bigr)$.

As the first step, we write the covariance as
\begin{equation*}
  \cov_{\rho}^\tau (A, B) 
  =
  \kw 2 \Tr\!\left( A^{(-)} \left(\rho^\tau\otimes \rho^\tau \right) 
    B^{(-)} \left(\rho^{1-\tau}\otimes \rho^{1-\tau} \right) \right) .
\end{equation*}
Using the swap-trick \eqref{eq:swap-trick} yields (see Figure~\ref{fig:swap})
\begin{equation} \label{eq:swaptrick}
  \cov_{\rho}^\tau (A, B)
  =
  \kw 2 \Tr\left(\swap^{1,3}\,\swap^{2,4}\,(A^{(-)} \otimes B^{(-)}) \, 
    \rho_4 \right) ,
\end{equation}
\begin{figure}[t]
\centering
\includegraphics{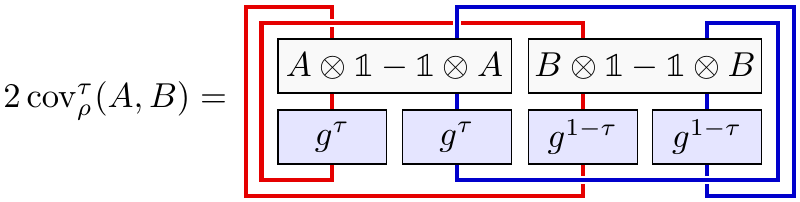}
\caption{The ``multiple swap-trick'': Eq.~\eqref{eq:swaptrick} as a tensor network.}
\label{fig:swap}
\end{figure}%
where
$\rho_4 \coloneqq \rho^\tau\otimes \rho^\tau \otimes \rho^{1-\tau}\otimes\rho^{1-\tau}$. 
For the case $\rho = \gibbs(\beta)$, the operator $\rho_4$ turns out to be
\begin{equation}
 \rho_4 = \frac{\e^{-\beta \tilde H^{(+)} }}{Z(\beta)^2} .
\end{equation}
Writing out $\rho_4$ as a power series yields
\begin{equation}\label{eq:cov_series}
 \cov_{\gibbs(\beta)}^\tau (A, B) 
  = 
  \kw{2\,Z(\beta)^2}\sum_{w \in \Eset^\ast} \frac{(-\beta)^{|w|}}{|w|!}\, t(w)
\end{equation}
with 
\begin{equation}\label{eq:trace_with_swaps}
 t(w)\coloneqq \Tr\!\left[\swap^{1,3}\,\swap^{2,4}\,(A^{(-)} \otimes B^{(-)}) \, \tilde h^{(+)}(w)\right]
\end{equation}
and $\tilde h^{(+)}(w) \coloneqq \tilde h^{(+)}_{w_1} \, \tilde h^{(+)}_{w_2} \dots \tilde h^{(+)}_{w_{|w|}}$. 
Next, we argue that $t(w)$ vanishes whenever $w$ does not contain a cluster connecting the supports of $A$ and $B$. 
Without loss of generality, we assume that 
$|\boundary\! A|\leq |\boundary\! B|$ 
and consider 
$\clusters[\boundary\! A]{\geq L}{}\complement = \Eset^\ast \setminus \clusters[\boundary\! A]{\geq L}{}$,
the set of words that do not contain a cluster containing an edge in $\boundary\! A$ of size $L \coloneqq \dist(A,B)$ or larger.
The set $\clusters[\boundary\! A]{\geq L}{}\complement$ hence contains no words with clusters that connect $\supp(A)$ and $\supp(B)$.
Any word $w \in \clusters[\boundary\! A]{\geq L}{}\complement$ can be replaced by a concatenation of two words $w_A$ and $w_B$ such that 
$\tilde h^{(+)}(w) = \tilde h^{(+)}(w_A) \, \tilde h^{(+)}(w_B)$, 
where $w_A$ contains all maximal clusters of $w$ that overlap with $\supp(A)$ and $w_B$ all other maximal clusters of $w$.
The operators $\tilde h^{(+)}(w_A)$ and $\1\otimes \1\otimes B^{(-)} \eqqcolon \hat B$, and $\tilde h^{(+)}(w_B)$ and $A^{(-)}\otimes\1\otimes \1 \eqqcolon \hat A$ then have disjoint supports, respectively, and the trace in Eq.~\eqref{eq:trace_with_swaps} factorizes into a product of two traces, one over the subsystem $X\coloneqq \supp(\hat A) \cup \supp(\tilde h^{(+)}(w_A))$ and the other over the rest of the system.
It turns out that both vanish: By using the symmetries 
$\hat A = - \swap^{1,2}\, \hat A\, \swap^{1,2}$, 
$\tilde h^{(+)}(w_A) = \swap^{1,2}\, \swap^{3,4}\, \tilde h^{(+)}(w_A)\, \swap^{3,4}\, \swap^{1,2}$, 
$\hat A \, \swap^{3,4} = \swap^{3,4}\, \hat A$, 
and that $\left(\swap^{i,j}\right)^2 = \1$, one can show, e.g., that
\begin{align} %\label{eq:exploitingthesymmetries}
     \Tr\bigl[\swap^{1,3}\,\swap^{2,4} \hat A\,\tilde h^{(+)}(w_A) \bigr] 
 = - \Tr\bigl[ \swap^{1,3}\, \swap^{2,4}\, \hat A\, \tilde h^{(+)}(w_A) \bigr] .
\end{align}
This implies that for every $w \in \clusters[\boundary\! A]{\geq L}{}\complement$,
\begin{equation}
  t(w) \propto \Tr\bigl[\swap^{1,3}\,\swap^{2,4}\hat A\,\tilde h^{(+)}(w_A) \bigr] = 0 \, .
\end{equation}
Together with Eq.~\eqref{eq:cov_series}, realizing that $Z(\beta)^2 = Z[H^{(+)}](\beta)$, and using the notation from Eq.~\eqref{eq:tildeOmega_def} with $F = \boundary\! A$ and $L = \dist(A,B)$, it follows that 
\begin{equation}
\cov^\tau_{\gibbs(\beta)}(A,B) 
= \Tr\biggl( \frac{\swap^{1,3}\swap^{2,4} \hat A \,\hat B} {2\, Z(\beta)^2} \,
       \Omega\bigl[\tilde H^{(+)}\bigr](\beta)\biggr) .
\end{equation}
After applying Lemma~\ref{lem:tailored_hastings} and using that
$\nnorm{\hat A}_\infty \leq 2 \norm{A}_\infty$, and similarly for $B$, we obtain
\begin{equation}\label{eq:cov_intermediate_bound_general}
\frac{|\cov_{\gibbs(\beta)}^\tau (A, B)| }{\norm{A}_\infty\, \norm{B}_\infty}
\leq 
 2\, \left(\e^{|\boundary\! A|\frac{\alphay(2\, \beta J)^L}{1-\alphay(2\, \beta J)} } - 1 \right).
\end{equation}
The fact that the condition $\beta < \beta^\ast$ is equivalent to $\alphay(2\, \beta  J)<1$ implies that $\alphay(2\, \beta J)^L$ decays exponentially with $L$. 
In order to obtain the desired exponential bound \eqref{eq:clustering}, we apply the bound 
$\forall x \in [0,x_0]: \exp(x)-1 \leq x\,(\e^{x_0} -1)/x_0$ with the choice $x_0 = \ln(3)$. 
In order to have $|\! \boundary\! A| \frac{\alphay(2\, \beta J)^L}{1-\alphay(2\, \beta  J)} \leq \ln(3)$, we impose 
\begin{align}
  L &\geq \left| \ln\left( \frac{|\boundary\! A|}{\ln(3)\, (1-\alphay(2\, \beta J))}\right) / \ln(2\, \beta \, J))\right| \\
  &= \xi(\beta)\, \left|\ln\!\left( {\ln(3)\, \bigl(1-\e^{-1/\xi(\beta)}\bigr)}/|\!\boundary\! A|\right) \right| \\ &\eqqcolon L_0(\beta,|\!\boundary\! A|) \, . \label{eq:L_zero_def} 
\end{align}
This guarantees the exponential bound \eqref{eq:clustering} and finishes the proof. 
\end{proof}

%%%=========================================================
%%%=====================  Fermions =========================
%%%=========================================================
\subsection{Proofs of the fermionic versions of the main results}
\label{sec:fermionic_proofs}
In order to also establish our main results for fermionic systems, we go through the proofs for spin systems and discuss the necessary modifications. 
 
\begin{proof}[Proof of Corollary~\ref{cor:truncation_formula_fermions}]
 In Theorem~\ref{thm:perturbation_formula}, we choose $H_0=H-H_{\boundary\! B}$. 
 As the local terms are all in $\G$, we have that the thermal state of $H_0$ factorizes, i.e., 
 $g_0 =  \gibbs[\trunc H B] \,  \gibbs[\trunc H {B\complement}]$. 
 After tracing over $B\complement$, the statement follows.
\end{proof}

\begin{proof}[Proof of Theorem~\ref{thm:clustering_fermions}]
  We use the same tensor copy trick as in the proof of Theorem~\ref{thm:clustering}. 
  Eq.~\eqref{eq:swaptrick} still holds in the fermionic setting. 
  Note that the Hilbert space over which the trace is performed in Eq.~\eqref{eq:swaptrick} is not the Fock space of a system of $4$ times the number of modes but the tensor product of four identical fermionic Fock spaces with the canonical inner product.
  This Hilbert space can be interpreted as that of a system of four types of fermionic particles that are each mutually indistinguishable and subject to (up to $\tau$-dependent prefactors) identical Hamiltonians but do not interact with each other and can be distinguished from each other.
  It is spanned by tensor products of Fock states.
  The state $\gibbs[\tilde H^{(+)}](\beta)$ is the thermal state of this system.
  Eq.~\eqref{eq:cov_series} with $t(w)$ as defined as in Eq.~\eqref{eq:trace_with_swaps} still holds.
  Note that the swap operators swap tensor factors, not fermionic modes. 
  Thus, they still satisfy the symmetry relations that are used to prove that only terms corresponding to words $w \in \clusters[\boundary\! A]{\geq L}{}$ can contribute to the covariance. 
 
  It remains to show that Lemma~\ref{lem:tailored_hastings} still holds in the fermionic setting.
  Lemmas~\ref{lem:alternating_binom_series} and \ref{lem:mfold_animals_to_animals} are purely combinatorial.
  Lemmas~\ref{lem:gt}, \ref{lem:eta_bound}, \ref{lem:rho_in_terms_of_eta}, \ref{lem:rho_of_G_tailored}, \ref{lem:rho_m_bound}, and \ref{lem:rho_m} only use the local boundedness of the Hamiltonian and that Hamiltonian terms with disjoint support commute.
  The same holds in the fermionic setting because the Hamiltonian terms must be physical operators, i.e., even polynomials in the fermionic operators.
  Hence all lemmas used in the proof of Lemma~\ref{lem:tailored_hastings} carry over to the fermionic setting.
  It is then straightforward to see that the proof itself also goes through without any modifications.
\end{proof}

\begin{proof}[Proof of Corollary~\ref{cor:intensivity_fermions}]
Tracing out $B\complement$ in the second trace in Eq.~\eqref{eq:truncation_error_in_terms_of_cov_fermions} and bounding the integral yields
\begin{equation}
  \begin{split}
    \big|\Tr[A\,\gibbs(\beta)]&-\Tr[A\,\gibbs[\trunc H B](\beta)]\big|\\ 
    &\leq |\beta |
    \sup_{s\in[0,1]} \sup_{\tau\in[0,1]} \bigl|\cov_{\gibbs_s(\beta)}^\tau(A, H_{\boundary\! B})\bigr| .
  \end{split}
\end{equation}
Taking the supremum over all $A$ with $\norm{A}_\infty=1$ and $\supp(A) \subseteq S$ and using  Theorem~\ref{thm:clustering_fermions} finish the proof. 
\end{proof}

%%%====================================================================
%%%======================     Conclusions     =========================
%%%====================================================================
\section{Conclusions}\label{sec:conclusions}
In this work, we clarify the limitations of a universal concept of scale-independent temperature by showing that temperature is intensive on a given length scale if and only if the correlations are negligible. 
The corresponding correlation measure turns out to also quantitatively capture the stability of thermal states against perturbations of the Hamiltonian. 
Moreover, we find a universal critical temperature above which correlations always decay exponentially with the distance. 
We compare our results to known results on phase transitions, comment on recent advances concerning thermalization in closed quantum systems (e.g., concerning the eigenstate thermalization hypothesis), and discuss known matrix product operator approximations of thermals states. 
More concretely, our results imply that at high enough temperatures, the error made when truncating a Hamiltonian at some distance away from the system of interest is exponentially suppressed with the distance. 
As a computational consequence, expectation values of local observables can be approximated efficiently. 

\section{Acknowledgments}
We thank I.~H.\ Kim, M.\ Holz\"{a}pfel, M.~B.\ Hastings, B.\ Nachtergaele, M.\ Friesdorf, and A.\ Werner for helpful feedback and discussions. 
This work was supported by the Studienstiftung des Deutschen Volkes, the Alexander von Humboldt Stiftung, the EU (Q-Essence, SIQS, RAQUEL), and the ERC (TAQ). 

%%%====================================================================
%%%======================     Appendix     =========================
%%%====================================================================
\section{Appendix: Cluster expansions and proof of Lemma~\ref{lem:tailored_hastings}}
\appendix
\setcounter{section}{1}
\setcounter{equation}{0}
\numberwithin{equation}{section}
The following discussion of cluster expansions is expected to be interesting in its own right, as it contains a rigorous formulation of the ideas outlined in Ref.~\cite{Hastings06}. 
We will provide a proof of the original statement used to establish Theorem~\ref{thm:MPO_approx} as well as of the tailored statement in Lemma~\ref{lem:tailored_hastings}, which is used to prove Theorem~\ref{thm:clustering} on the clustering of correlations. 

\subsection{The original cluster expansion from Ref.~\texorpdfstring{\cite{Hastings06}}{}}
% \subsection{The original cluster expansion from Ref.~\cite{Hastings06}}
The original cluster expansion is similar to Lemma~\ref{lem:tailored_hastings} with just one copy of the system instead of two weighted ones: 

\begin{lemma}[Truncated cluster expansion \cite{Hastings06}]\label{lem:hastings}
Let $H=\sum_{\lambda \in \Eset} h_\lambda$ be a local Hamiltonian with finite interaction graph $(V,\Eset)$ having growth constant $\animalc$ and local interaction strength  
$J = \max_{\lambda \in \Eset} \norm{h_\lambda}_\infty$, 
and define 
$\alphay(x) \coloneqq \animalc \, \e^{|x|}\,\left(\e^{|x|} -1\right)$.
Moreover, let $\beta$ be small enough such that $\alphay(\beta J)<1$.
For some subset of edges $\Fset \subset \Eset$ let $\clusters{\geq L}{}\subset \Eset^\ast$ be the set of words containing at least one cluster $c$ that contains at least one letter of $\Fset$ and has size $|c|\geq L$ and denote the corresponding truncated cluster expansion by 
\begin{equation}\label{eq:Omega_def}
 \Omega[H](\beta) \coloneqq \sum_{w\in \clusters{\geq L}{}} \frac{(-\beta)^{|w|}}{|w|!} h(w) \, .
\end{equation}
Then, 
\begin{equation} \label{eq:Hastings_bound}
 \frac{\norm{\Omega[H](\beta)}_1}{Z(\beta)}
 \leq 
 \exp\left(|\Fset|\frac{\alphay(\beta J)^L}{1-\alphay(\beta J)} \right) - 1 .
\end{equation}
\end{lemma}

If one applies this lemma to the setting of Lemma~\ref{lem:tailored_hastings}, one obtains a bound similar as the one in Eq.~\eqref{eq:tailored_Hastings_bound} but with $Z[\tilde H](\beta)$ instead of $Z(\beta)$, where the ratio $Z[\tilde H](\beta) / Z[H](\beta)$ can be exponentially large in the system size for $\tau \in ]0,1[$. 

Lemma~\ref{lem:hastings} was used in Ref.~\cite{Hastings06} to establish a mathematically (not algorithmically) constructive version of Theorem~\ref{thm:MPO_approx}, on MPO approximations, where the MPO in Eq.~\eqref{eq:MPO_approx_error} is given by 
\begin{equation}\label{eq:MPO_approx}
 \rho(\beta, L) =\kw{Z(\beta)} \sum_{w \in E^\ast \setminus \clusters[E]{\geq L}{}} \frac{(-\beta)^{|w|}}{|w|!}h(w) \, .
\end{equation}

\subsection{Proofs of Lemmas~\ref{lem:tailored_hastings} and~\ref{lem:hastings}}
The purpose of this section is to prove Lemma~\ref{lem:tailored_hastings}. 
But, along the way, we also prove Lemma~\ref{lem:hastings}. 
In order to do so, we start with the introduction of some more notation, mainly concerning clusters and lattice animals. 
For $w \in E^\ast$ and any sub-alphabet $G\subset E$, we write $G \subset w$ if every letter in $G$ also occurs in $w$. 
By $G\complement\coloneqq E\setminus G$, we denote the \emph{complement} of $G \subset E$. 
The \emph{extension} of $G$ is defined to be
$\ol G \coloneqq \{ \lambda \in \Eset \mid \exists \lambda' \in G : \lambda' \cap \lambda \neq \emptyset\}$ and, similarly as for subsystems, its \emph{boundary} is $\boundary\! G \coloneqq \ol G \setminus G$.
Throughout the proof, we fix some subset of edges $F \subset E$. 
We denote by $\clusters{\geq L}{} \subset E^\ast$ the set of words that contain at least one cluster $c$ 
with $c \cap F \neq \emptyset$ and $|c|\geq L$, and we denote by $\clusters{\geq L}{k}$ the set of words that contain exactly $k$ such clusters.
Note that for an \emph{animal} $G \subset E$, there exists a cluster $c \in E^\ast$ such that 
$G = \{\lambda \in c\}$, and if one imposes some order on $G$, one obtains a cluster.
We denote by $\animals{=l}{}$ and $\animals{\geq L}{}$ the sets of animals that contain at least one edge of $F$ and are of size exactly $l$ or at least $L$, respectively. 
Moreover, we denote by $\animals{\geq L}{k}$ the corresponding sets of \emph{$k$-fold animals}, i.e.,
\begin{equation*}
    \animals{\geq L}{k}  \coloneqq 
    \left\{\dUnion_{j=1}^k G_j: G_j \in \animals{\geq L}{} \text{ non-overlapping} \right\}.
\end{equation*}
For a more compact notation, we write the terms in the exponential series as
 \begin{equation}\label{eq:f_def}
  f(w) \coloneqq \frac{(-\beta)^{|w|}}{|w|!} h(w) \, .
\end{equation}
We will frequently use the following fact: 
For any Hamiltonian with a finite interaction graph $(V,E)$, the partial series over any set of words $\W \subseteq E^\ast$ converges absolutely, i.e., 
\begin{align}
 \Bnorm{\sum_{w \in \W} f(w)}_\infty
 &\leq \sum_{w \in \W} \frac{(|\beta|\,J)^{|w|}}{|w|!} 
 \\
 &\leq \sum_{w \in E^\ast} \frac{(|\beta|\,J)^{|w|}}{|w|!} 
 \\
 &= \exp(|\beta|\,J\,|E|) \, .
\end{align}
In particular, this bound implies that the order of the terms in the series over any subset of words $\W$ 
does not matter.

%%% --------------------------------------------------------
%%% --------------------- main proof -----------------------
%%% --------------------------------------------------------
In the following proofs of Lemmas~\ref{lem:tailored_hastings} and~\ref{lem:hastings} we use several technical auxiliary lemmas, which we will only state and prove subsequently. 

\begin{proof}[Proof of Lemma~\ref{lem:tailored_hastings}]
During this proof, we indicate quantities corresponding to $\tilde H$ by a tilde accent, e.g., $\tilde f(w)$ is defined as in Eq.~\eqref{eq:f_def} but with respect to the local terms $\tilde h_\lambda$ of $\tilde H$ while $f(w)$ is defined with respect to the local terms $h_\lambda$ of $H$. 

We start the proof by rearranging the terms in the series over $\clusters{\geq L}{}$ in Eq.~\eqref{eq:tildeOmega_def} according to the number of relevant clusters they contain and use Lemma~\ref{lem:alternating_binom_series} with $b_k$ being the series over 
$\clusters{\geq L}{k}$ to obtain 
\begin{align}
 \Omega[\tilde H](\beta)
 &= \sum_{k=1}^\infty \sum_{w \in \clusters{\geq L}{k}}\tilde f(w)  \label{eq:use_alternating_binom_series}
 \\
 &= -\sum_{m=1}^\infty (-1)^m \sum_{k=m}^\infty \binom k m 
    \sum_{ w \in \clusters{\geq L}{k} } \tilde f(w) .
 \nonumber
\end{align}
Lemmas~\ref{lem:eta_bound}, \ref{lem:rho_in_terms_of_eta}, and \ref{lem:rho_m_bound} are the core of the proof. They define a series of operators $(\tilde \rho_m)_{m=1}^\infty$ that have a particularly useful form given in Lemma~\ref{lem:rho_m}. 
This form exactly matches the series over $k$ in Eq.~\eqref{eq:use_alternating_binom_series}, which leads to the following identity:
\begin{equation} \label{eq:trunc_series_in_terms_rho_m}
 \Omega[\tilde H](\beta)
 = 
 -\sum_{m=1}^\infty (-1)^m  \tilde \rho_m \, .
\end{equation} 
The operators $\tilde \rho_m$ are defined in Eq.~\eqref{eq:rho_m_in_terms_rhoF} as series over $m$-fold lattice animals $G$ of operators $\rho(G)$ [defined in Eq.~\eqref{eq:rhoF_in_terms_eta}]. 
This yields
\begin{equation}\label{eq:Tr_SABrhom}
\Tr(\swap \, A \, B \, \tilde \rho_m) 
= 
\sum_{G \in \animals{\geq L}{m}} \Tr( \swap \, A \, B \, \tilde \rho(G))\, . 
\end{equation}
In the previous steps, the series over words has been rewritten as a series over $m$-fold animals. 
Lemma~\ref{lem:rho_of_G_tailored} provides a bound on $\tilde \rho(G)$ that, together with Eqs.~\eqref{eq:trunc_series_in_terms_rho_m} and~\eqref{eq:Tr_SABrhom}, yields
\begin{equation}
\frac{\bigl| \Tr\bigl( \swap \, A \, B\; \Omega[\tilde H](\beta) \bigr) \bigr|}
	{\norm{A}_\infty \norm{B}_\infty Z(\beta)}
\leq 
\sum_{m=1}^\infty \, \sum_{G\in \animals{\geq L}{m}} y(\beta J)^{|G|} .
\end{equation}
Now, a counting argument for lattice animals from Lemma~\ref{lem:mfold_animals_to_animals} allows us to bound the series over $m$-fold animals $G$ in terms of a series of animals
\begin{equation*}
\frac{\bigl| \Tr\bigl( \swap \, A \, B \; \Omega[\tilde H](\beta) \bigr) \bigr|}
	{\norm{A}_\infty \norm{B}_\infty Z(\beta)}
 \leq 
\sum_{m=1}^\infty \kw{m!} \left(\sum_{G \in \animals{\geq L}{}} y(\beta J)^{|G|} \right)^m .
\end{equation*}
Using that the number $a_l$ [see Eq.~\eqref{eq:animal_bound}] of lattice animals $G$ with $G \cap F\neq \emptyset$ and of size $|G|=l$ is bounded by $|F|\,a_l$ and that $a_l \leq \animalc^l$ [see Eq.~\eqref{eq:animal_const}], we obtain 
\begin{equation*} 
 \bigl| \Tr\bigl( \swap \, A \, B \; \Omega[\tilde H](\beta) \bigr) \bigr|  
 \leq 
 Z(\beta)\sum_{m=1}^\infty \kw{m!} \left(|F| \sum_{l=L}^\infty \alphay(\beta J)^l \right)^m 
\end{equation*}
with $\alphay(x) \coloneqq \animalc \, y(x)$. 
Performing the partial geometric series over $l$ with argument  
$\alphay(\beta J) <1$ and the exponential series over $m$ yields Eq.~\eqref{eq:Hastings_bound}
and completes the proof.
\end{proof}

Similarly, we prove Lemma~\ref{lem:hastings}:
\begin{proof}[Proof of Lemma~\ref{lem:hastings}]
By the same argument that led us to Eq.~\eqref{eq:trunc_series_in_terms_rho_m} in the proof of Lemma~\ref{lem:tailored_hastings}, we obtain
\begin{equation} 
 \Omega[H](\beta)
 = 
 -\sum_{m=1}^\infty (-1)^m  \rho_m \, .
\end{equation} 
Applying the triangle inequality and using the bound on $\rho_m$ from Lemma~\ref{lem:rho_m_bound} yields
\begin{equation}
 \norm{\Omega[H](\beta)}_1 
 \leq Z(\beta)\sum_{m=1}^\infty \kw{m!} \left(|F| \sum_{l=L}^\infty \alphay(\beta J)^l \right)^m .
\end{equation}
Performing the partial geometric series over $l$ with argument  
$\alphay(\beta J) <1$ and the exponential series over $m$ yields Eq.~\eqref{eq:Hastings_bound}
and completes the proof.
\end{proof}

%%% --------------------------------------------------------
%%% ----------------------- Lemmas -------------------------
%%% --------------------------------------------------------
We now prove various lemmas that are used in the previous proofs of Lemmas~\ref{lem:tailored_hastings} and~\ref{lem:hastings}.

% ---------------------- binomial sum lemma ----------------------------
\begin{lemma}\label{lem:alternating_binom_series}
Let $(b_k)_{k=1}^\infty$ be a sequence of complex matrices 
\begin{align}
A_K &\coloneqq \sum_{k=1}^K b_k 
\intertext{and}
B_K &\coloneqq -\sum_{m=1}^K (-1)^m \sum_{k=m}^K  \binom k m \,b_k \, .
\end{align}
Then, $A_K = B_K$ for all $K \in \NN$. 
In particular, if both sequences converge, then their limits are the same, i.e., 
$\lim_{K \to \infty} A_K = \lim_{K \to \infty} B_K$.
\end{lemma}

\begin{proof}
Applying the binomial theorem to $(1-1)^k = 0$ yields 
\begin{equation}\label{eq:binomial_alternating_sum}
\sum_{l=0}^k (-1)^l \,\binom k l = 0 \ ,
\end{equation}
which we will use. 
We prove the identity by induction. 
$A_1 = B_1$ is easy to see. 
Under the assumption that $A_K=B_K$ for some $K\in \NN$, we obtain 
\begin{align}
B_{K+1}&= B_K - (-1)^{K+1} \,\binom{K+1}{K+1}\,b_{K+1} 
\\
&\phantom{{}= B_K } - \sum_{m=1}^{K} (-1)^m \,\binom{K+1}{m}\,b_{K+1} 
 %phantom in an align needs the {}
\\
&=A_K +\left( - (-1)^{K+1} - \sum_{m=1}^{K} (-1)^m\,\binom{K+1}{m}\right) b_{K+1}\nonumber \\
&= A_{K+1} \, ,
\end{align}
where we have used Eq.~\eqref{eq:binomial_alternating_sum} in the last step.
This proves the lemma.
\end{proof}

The goal of the following lemmas is to show that $\rho_m$ is well-defined and to upper bound it in $1$-norm. 
The order of the lemmas is chosen in a way that makes clear that the two quantities $\rho_m$ and $\rho(G)$, which will be defined shortly, are actually well-defined.

%%% ----------------------------- Golden Thomson lemma --------------------------------
We start with a $1$-norm bound on the perturbed exponential series. 
\begin{lemma}[Eq.~(21) from Ref.~\cite{Hastings06}]\label{lem:gt}
Let $H$ be a Hamiltonian with finite interaction graph $(V,E)$.  
For any sequence $(G_j)_{j=1}^k$ of sub-alphabets $G_j \subset E$, 
  \begin{equation}
    \bnorm{\e^{-\beta\,(H- \sum_{j=1}^k H_{G_j})}}_1 
    \leq 
    Z(\beta) \prod_{j=1}^k \bnorm{\e^{|\beta| \,H_{G_j}}}_\infty .
  \end{equation}
\end{lemma}

\begin{proof}
The lemma is essentially a consequence of the Gol\-den-Thompson inequality and the fact that the $1$-norm of a positive operator coincides with its trace. Using first the Golden-Thompson and then H\"older's inequality, we obtain
  \begin{align}
    \bnorm{\e^{-\beta\,(H- \sum_{j=1}^k H_{G_j})}}_1
    \leq 
    \Tr\bigl[ \e^{-\beta\,( H-\sum_{j=1}^{k-1}H_{G_j})} &\e^{\beta \,H_{G_k}}\bigr]
    \nonumber \\
    \leq  
    \Tr\bigl[ \e^{-\beta\,( H-\sum_{j=1}^{k-1}H_{G_j})}\bigr] \bnorm{ \e^{|\beta|\, H_{G_k}}}_\infty& .
   \end{align}
Iteration completes the proof.
\end{proof}

%%% ----------------- lem: eta --------------------
We will use the following lemma to bound the operator norm of certain subseries of $f(w)$. 

\begin{lemma}\label{lem:eta_bound}
Let $(V,E)$ be a finite graph and $J\geq 0$. 
For any $G\subset E$, 
\begin{equation}\label{eq:eta_bound_new}
  \sum_{w \in G ^\ast: G \subset w} 
      \frac{|\beta J|^{|w|}}{|w|!} 
=
\bigl(\e^{|\beta J|}-1\bigr)^{|G|} .
\end{equation}
\end{lemma}

\begin{proof}
Ordering the words in the sum in Eq.~\eqref{eq:eta_bound_new} with respect to their length yields
 \begin{align}
\sum_{w \in G ^\ast: G \subset w} 
      \frac{|\beta J|^{|w|}}{|w|!}
&=
 \sum_{l=|G|}^\infty \sum_{w \in G^l : G\subset w} \frac{|\beta J|^{|w|}}{|w|!} 
\\
&= 
\sum_{l=|G|}^\infty \frac{|\beta J|^l}{l!} |\{w \in G^l : G\subset w\}| \, . \label{eq:intermediate_bound_eta}
\end{align}
From basic combinatorial considerations, we obtain
\begin{equation}
|\{w \in G^l : G\subset w\}| = \sum_{\substack{j_1, j_2, \dots, j_n \geq 1,\\ j_1+j_2+\dots+j_n = l}} \binom l j \, ,
\end{equation}
where $\binom l j$ is a multinomial coefficient. 
Therefore, the right-hand side of Eq.~\eqref{eq:intermediate_bound_eta} only depends on $n\coloneqq|G|$ and we denote it by
\begin{align}
 \gamma(n) &\coloneqq \sum_{l=n}^\infty \gamma(n,l)
 \\
 \intertext{with}
 \gamma(n,l)&\coloneqq \frac{|\beta J|^l}{l!} \sum_{\substack{j_1, j_2, \dots, j_n \geq 1,
 \\ j_1+j_2+\dots+j_n = l}} \binom l j
 \\
 &= \sum_{\substack{j_1, j_2, \dots, j_n \geq 1,\\ j_1+j_2+\dots+j_n = l}} 
	\frac{|\beta J|^{j_1}}{j_1!} \frac{|\beta J|^{j_2}}{j_2!} \dots \frac{|\beta J|^{j_n}}{j_n!}
 \, . \nonumber
\end{align}
Then, 
\begin{align}
\gamma(n)
&= 
\sum_{l=n}^\infty \sum_{\substack{j_1, j_2, \dots, j_n \geq 1,\\ j_1+j_2+\dots+j_n = l}} 
      \frac{|\beta J|^{j_1}}{j_1!} \frac{|\beta J|^{j_2}}{j_2!} \dots \frac{|\beta J|^{j_n}}{j_n!}
      \nonumber \\
&= 
\sum_{l=n}^\infty \sum_{j_1=1}^{l-(n-1)} \frac{|\beta J|^{j_1}}{j_1!}
	\quad \sum_{\mathclap{\substack{j_2,\dots,j_n \geq 1,\\ j_2+\dots+j_n = l-j_1}}}\quad \frac{|\beta J|^{j_2}}{j_2!} \dots \frac{|\beta J|^{j_n}}{j_n!} \nonumber \\
&=
\sum_{l=1}^\infty \sum_{j_1=1}^{l} \frac{|\beta J|^{j_1}}{j_1!} \gamma(n-1,l+n-1-j_1)
\end{align}
and, after realizing that the last series is a Cauchy product,
\begin{align}
 \gamma(n) &= 
 \sum_{j_1=1}^\infty \frac{|\beta J|^{j_1}}{j_1!} \sum_{l=n-1}^\infty \gamma(n-1,l)\\
 &= (\e^{|\beta J|} -1)\, \gamma(n-1) \, . 
\end{align}
We note that $\gamma(1) = \e^{|\beta J|} -1$, and iteration finishes the proof.
\end{proof}

The following lemma provides a factorization of the series $\rho(G)$ in Eq.~\eqref{eq:rhoF_def} over words that have no letters on the boundary of an $m$-fold animal $G\in \animals{=l}{m}$ and contain all letters in $G$, into $\exp(-\beta\, H_{(\ol G)\complement})$, whose norm we have bounded in Lemma~\ref{lem:gt}, times a product of operators $\eta(G_j)$. 
The $\eta(G_j)$ are supported on the single animals $G_j$ composing the $m$-fold animal $G$. 
As we will see, a norm bound for  $\eta(G_j)$ follows immediately from the previous lemma, which, in turn, also yields an upper bound on $\rho(G)$. The form of $\rho(G)$ given in Eq.~\eqref{eq:rhoF_def} together with this upper bound plays an important role in the main cluster expansion. 

\begin{lemma}\label{lem:rho_in_terms_of_eta}
Let $H$ be a Hamiltonian with finite interaction graph $(V,E)$. 
For $G\subset E$, let $G = \dUnion_{j=1}^m G_j$ be the decomposition of $G$ into non-overlapping animals $G_j \subset E$ and define
\begin{equation} 
\rho(G) \coloneqq \e^{-\beta\,  H_{(\ol G)\complement}} \prod_{j=1}^m \eta(G_j)\,  \label{eq:rhoF_in_terms_eta}
\end{equation}
with
\begin{equation} \label{eq:eta_def_new}
\eta(G) \coloneqq \sum_{w \in G ^\ast: G \subset w} f(w) \, .
\end{equation}
Then, 
\begin{equation}\label{eq:rhoF_def}
 \rho(G) = \sum_{w \in [(\boundary G)\complement]^\ast: \, G \subset w} f(w) \, .
\end{equation}
\end{lemma}

\begin{proof}
To simplify the notation, we denote the relevant set of words that contain no letters in 
$\boundary\! G$ and each letter in $G$ at least once by
\begin{equation}
 \W^{\supset G} \coloneqq \{w \in [(\boundary G)\complement]^\ast: G\subset w\} \, .
\end{equation}
The idea is to group these words into subsets 
$[w]\subset \W^{\supset G}$ that coincide on the connected components of $G$ and on $(\ol G)\complement$ and correspondingly split up the series \eqref{eq:rhoF_def}. 
We formalize this idea by introducing an equivalence relation on $\W^{\supset G}$. 
For $v,w \in \W$, we define
\begin{equation*}
v \sim w \ \colonequiv \, 
\begin{cases}
 v\restr G\complement = w\restr G\complement \\
 v\restr G_j = w\restr G_j \quad \forall j=1,2, \dots, k
\end{cases},
\end{equation*}
where, for any sub-alphabet $G' \subset E$, the restriction $w\restr G'$ of a word $w \in E^\ast$ is obtained from $w$ by omitting all letters that are not in $G'$. 
Then, the size of each equivalence class $[w] \in \W^{\supset G}/\! \sim$ is given by the multinomial coefficient 
\begin{equation}
|[w]|=\binom{|w|}{(|w\restr G\complement|, |w\restr G_1|, \ldots, |w\restr G_k|)} .
\end{equation}
Note also that $h([w])\coloneqq h(w\restr \ol G\complement)\prod_{j=1}^k h(w\restr G_j) = h(w)$ is well-defined as a function on the classes. 
Let us denote the set of words over the alphabet $G_j$ that contain all letters at least once by
\begin{equation}
  \W^{=G_j} \coloneqq \{w \in (G_j)^\ast: G_j \subset w\} \, .
\end{equation}
Then, the quotient set can be identified with a Cartesian product of these sets 
\begin{equation}
  \W^{\supset G}/\!\! \sim\ \cong [(\ol G)\complement]^\ast\times \bigtimes_{j=1}^k \W^{=G_j} .
\end{equation}
\begin{widetext}
\noindent For each equivalence class $K \in \W^{\supset G}/\!\! \sim$ we pick an arbitrary representative $w_K \in \W^{\supset G}$, use the definition of $f$ in Eq.~\eqref{eq:f_def}, and determine that $k$ is the number of connected components of $G$. This yields
\begin{align}
\sum_{w \in [(\boundary G)\complement]^\ast: \, G \subset w} f(w) 
&= 
\sum_{K \in \W^{\supset G} /\! \sim} |K| \frac{(-\beta)^{|w_K|}}{|w_K|!} h(w_K)\\
&= 
\sum_{v \in [(\ol G)\complement]^\ast} 
\sum_{w_1 \in  \W^{=G_1} } \sum_{w_2 \in  \W^{=G_2} } \dots \sum_{w_k \in  \W^{=G_k} }
\binom{|v|+\sum_{j=1}^k|w_j|}{(|v|, |w_1|, \ldots, |w_k|)} \frac{(-\beta)^{|v|+\sum_{j=1}^k|w_j|}}{(|v|+\sum_{j=1}^k|w_j|)!} h(v)\prod_{j=1}^k h(w_j) \nonumber \\ 
&=\sum_{v \in [(\ol G)\complement]^\ast} f(v) 
\left(\sum_{w_1 \in  \W^{=G_1} } f(w_1) \right)
\left(\sum_{w_2 \in  \W^{=G_2} } f(w_2) \right) \dots
\left(\sum_{w_k \in  \W^{=G_k} } f(w_k) \right)\, .
\end{align}
\end{widetext}
Using the definition of $\eta$ from Eq.~\eqref{eq:eta_def_new} on the last factors yields
\begin{equation}
\sum_{w \in [(\boundary G)\complement]^\ast: \, G \subset w} f(w) 
= 
\e^{-\beta H_{(\ol G)\complement}}\prod_{j=1}^k \eta(G_j) = \rho(G)  .
\end{equation}
\end{proof}

%%% ----------------- lem: decompose eta to bound trace -----------------
The following lemma is a tighter variant of some of the original arguments leading to Lemma~\ref{lem:hastings} for Hamiltonians consisting of two weighted copies of a local Hamiltonian. 
Its purpose is to provide a specialized tighter bound on $\rho(G)$, which turns out to be sufficient for our purposes. 
The central idea of the lemma is to expand $\rho(G)$ in the left-hand side of Eq.~\eqref{eq:Tr_eta_bound} in order to be able to bound the trace using the generalized H\"older's inequality. 

\begin{lemma}\label{lem:rho_of_G_tailored}
Let $\tau$, $H$, $\tilde H$, $A$, and $B$ be as in Lemma~\ref{lem:tailored_hastings} and let $G \in \animals{\geq L}{=m}$ be an $m$-fold lattice animal with $G = \bigcup_{j=1}^m G_j$ and $G_j \in \animals{\geq L}{}$. 
Moreover, let $\tilde \rho(G)$ be defined as $\rho(G)$ in Eq.~\eqref{eq:rhoF_in_terms_eta} but with respect to $\tilde H$. 
Then, 
\begin{equation} \label{eq:Tr_eta_bound}
  \frac{\bigl| \Tr\bigl[ \swap A^\I  B^\II\, \e^{-\beta \tilde H_{\ol{G}\complement}} \tilde \rho(G) \bigr] \bigr|} 
       {\norm{A}_\infty \norm{B}_\infty Z(\beta)}
  \leq
  y(\beta J)^{|G|} ,
\end{equation}
where $y(x) \coloneqq \e^{|x|}\bigl(\e^{|x|} -1\bigr)$.
\end{lemma}

\begin{proof}
Let us denote $k_\lambda^\I \coloneqq \tau h_\lambda^\I$ and 
$k_\lambda^\II \coloneqq (1-\tau) h_\lambda^\II$.
For $w \in E^\ast$ and $v \in \{1,2\}^{|w|}$, we define 
$\tilde h(w,v) \coloneqq k^{v_1}_{w_1} k^{v_2}_{w_2} \dots k^{v_{|w|}}_{w_{|w|}}$. 
Then, by expanding the product $\tilde h(w)$, it can be written as 
\begin{equation}
\tilde h(w) = \sum_{v \in \{1,2\}^{|w|}} \tilde h(w,v) . 
\end{equation}
Importantly, we can reorder the terms in $\tilde h(w,v)$ so that 
\begin{equation}
 \tilde h(w,v) = \tilde h^\I(w,v) \tilde h^\II(w,v)\, ,
\end{equation}
where 
$\tilde h^\I(w,v) = h^{(i)}(w,v)\otimes \1$ and 
$\tilde h^\II(w,v)= \1 \otimes h^{(ii)}(w,v)$. 
Factorizing the operators and using the swap-trick \eqref{eq:swap-trick}, we obtain 
\begin{widetext}
\begin{align}
  \Tr\bigl[ \swap A^\I  B^\II  
  \e^{-\beta \tilde H_{\ol{G}\complement}}
  \tilde h(w,v) \bigr] 
  &= 
  \Tr\bigl(\swap 
    \bigl[A \, \e^{-\beta \tau H_{\ol{G}\complement}} h^{(i)}(w,v)\bigr] \otimes
    \bigl[B \, \e^{-\beta (1-\tau) H_{\ol{G}\complement}} h^{(ii)}(w,v) \bigr]
  \bigr) 
  \\
  &=
  \Tr\bigl( 
    \bigl[A \, \e^{-\beta \tau H_{\ol{G}\complement}} h^{(i)}(w,v)\bigr] \,
    \bigl[B \, \e^{-\beta (1-\tau) H_{\ol{G}\complement}} h^{(ii)}(w,v) \bigr]
  \bigr) . 
\end{align}
Bounding the trace by the $1$-norm and applying H\"{o}lder's inequality generalized to several operators yields
\begin{align}
  \bigl| \Tr\bigl[ \swap A^\I  B^\II  
  \e^{-\beta \tilde H_{\ol{G}\complement}}
  \tilde h(w,v) \bigr] \bigr|
  &\leq \norm{A}_\infty \norm{B}_\infty 
  \bnorm{\e^{-\beta \tau H_{\ol{G}\complement}}}_{1/\tau}
  \bnorm{\e^{-\beta (1-\tau) H_{\ol{G}\complement}}}_{1/(1-\tau)}
  \bnorm{h^{(i)}(w,v)}_\infty \bnorm{h^{(ii)}(w,v)}_\infty \nonumber
  \\
  &\leq \norm{A}_\infty \norm{B}_\infty \norm{\e^{-\beta H_{\ol{G}\complement}}}_1 
    \, J^{|w|} \tau^{n^\I(v)}\, (1-\tau)^{n^\II(v)}, 
\end{align}
where in the second step, we have used that $\norm{X}_p = \norm{|X|^p}_1^{1/p}$ and that with $n^{(j)}(v) \coloneqq |\{ v_k : v_k = j \}|$ for $j\in \{1,2\}$ the bounds
$\nnorm{h^{(i)}(w,v)}_\infty \leq (\tau\, J)^{n^\I(v)}$ and 
$\nnorm{h^{(ii)}(w,v)}_\infty \leq ((1-\tau)\, J)^{n^\II(v)}$ hold.
Now, we apply Lemma~\ref{lem:gt} and use that 
$\nnorm{e^{|\beta| H_{G_j}}}_\infty \leq \e^{|\beta|J |G_j|}$ to arrive at
\begin{equation}
  \bigl| \Tr\bigl[ \swap A^\I  B^\II \, 
  \e^{-\beta \tilde H_{\ol{G}\complement}}
  \tilde h(w,v) \bigr] \bigr|
  \leq 
  \norm{A}_\infty \norm{B}_\infty
  Z(\beta) \,
  \e^{|\beta| \, J \, |G|}
  J^{|w|} \, \tau^{n^\I(v)} \, (1-\tau)^{n^\II(v)} . \label{eq:bound_Tr_with_hwv}
\end{equation}
From the definition of $\eta$ in Eq.~\eqref{eq:eta_def_new}, it follows that
\begin{align}
 \prod_{j=1}^m \tilde{\eta}(G_j)
 &= 
 \sum_{\left\{\substack{w^{(j)} \in G_j^\ast:\\ G_j \subset w^{(j)} }\right\}_{j=1}^m }
 \prod_{i=1}^m
 \frac{(-\beta)^{\left|w^{(i)}\right|}}
      {  \left|w^{(i)}\right| !}\,
  \tilde h\bigl(w^{(i)} \bigr)\\
 &= 
 \sum_{\left\{\substack{w^{(j)} \in G_j^\ast:\\ G_j \subset w^{(j)} }\right\}_{j=1}^m }
 \frac{(-\beta)^{|w|}}
      { \prod_{i=1}^m \left|w^{(i)}\right| !}\,
 \sum_{ v\in \{1,2\}^{|w|}}
     \tilde h(w,v) \, ,
\end{align}
where $w \coloneqq w^\I w^\II \dots w^{(m)}$ and hence 
$\tilde h(w) = \prod_{i=1}^{m} h\bigl(w^{(i)}\bigr)$. 
Together with the bound \eqref{eq:bound_Tr_with_hwv}, this yields
\begin{equation}
\frac{\bigl| \Tr\bigl[ \swap A^\I  B^\II  
  \e^{-\beta \tilde H_{\ol{G}\complement}}
  \prod_{j=1}^m \tilde{\eta}(G_j) \bigr] \bigr|}
     {\norm{A}_\infty \norm{B}_\infty Z(\beta)}
  \leq 
  \e^{|\beta| \, J \, |G|}
   \sum_{\left\{\substack{w^{(j)} \in G_j^\ast:\\ G_j \subset w^{(j)} }\right\}_{j=1}^m }
   \frac{|\beta|^{|w|}}
	{ \prod_{i=1}^m \left|w^{(i)}\right| !} \,
      J^{|w|} \sum_{v \in \{1,2\}^{|w|}} \tau^{n^{(1)}(v)} \, (1-\tau)^{n^{(2)}(v)} .
\end{equation}
\end{widetext}
Using the definition \eqref{eq:rhoF_in_terms_eta} of $\tilde \rho(G)$ and the multinomial formula yields
\begin{align}
& \phantom{={}}
  \frac{\bigl| \Tr\bigl( \swap A^\I  B^\II  
  \e^{-\beta \tilde H_{\ol{G}\complement}}
  \tilde \rho(G) \bigr) \bigr|}
     {\norm{A}_\infty \norm{B}_\infty Z(\beta)}
   \\
& = 
  \e^{|\beta| \, J \, |G|}
   \sum_{\left\{\substack{w^{(j)} \in G_j^\ast:\\ G_j \subset w^{(j)} }\right\}_{j=1}^m }
   \prod_{i=1}^m \frac{(|\beta|\, J)^{|w^{(i)}|}}
	{\left|w^{(i)}\right| !} 
  \\
& =  
  \e^{|\beta| \, J \, |G|}
    \prod_{i=1}^m \Biggl(\, \sum_{\substack{w^{(i)} \in G_i^\ast: \\ G_i \subset w^{(i)}}}
      \frac{ ( |\beta| J)^{\bigl|w^{(i)}\bigr|} }{ \bigl|w^{(i)}\bigr|! } \Biggr)
   \\
& \leq 
  \e^{|\beta| \, J \, |G|}
    \prod_{i=1}^m \bigl(\e^{ |\beta|\, J} - 1\bigr)^{|G_i|} ,  \label{eq:Tr_eta_bound_intermediate}
\end{align}
where in the second to last step, we have factorized the series and in the last step, we have used Lemma~\ref{lem:eta_bound}.
\end{proof}

%%% ----------------- lem: m-fold animals to animals --------------------
We will need the following combinatorial lemma:
\begin{lemma}\label{lem:mfold_animals_to_animals}
Let $(V,E)$ be a finite (hyper)graph and $y \in [0,1[$. Then, for any $F \subset E$,
\begin{equation}
  \sum_{G \in \animals{\geq L}{m}} y^{|G|} \leq \kw{m!} \left(\sum_{G \in \animals{\geq L}{}} y^{|G|} \right)^m .
\end{equation}
\end{lemma}

\begin{proof}
Remember that $\animals{\geq L}{m}$ is the set of $m$-fold (edge) animals of size at least $L$ that contain a letter from $F$.
For every $G \in \animals{\geq L}{m}$, one finds $m$ pairs $(G_1,G_2)$ with $G_1 \in \animals{\geq L}{m-1}$ and $G_2 \in \animals{\geq L}{}$ such that $G= G_1 \dunion G_2$; hence,
\begin{align}
m \sum_{G \in \animals{\geq L}{m}} y^{|G |}
&\leq 
\sum_{G_1 \in \animals{\geq L}{m-1}} \, \sum_{G_2 \in \animals{\geq L}{}} y^{|G_1 |+|G_2 |} 
\nonumber \\
&= \left( \sum_{G \in \animals{\geq L}{m-1}} y^{|G |}\right)
 \left(\sum_{G \in \animals{\geq L}{}} y^{|G |}\right).\nonumber
\end{align}
By iterating this inequality, we obtain
\begin{equation}
\sum_{G \in \animals{\geq L}{m}} y^{|G |} 
\leq
\kw{m!}  \left(\sum_{G \in \animals{\geq L}{}} y^{|G |}\right)^m .
\end{equation}
\end{proof}

% ---------------------- lemma: \rho_m bound ----------------------------
In the following lemma, we define a family of operators $\rho_m$ and bound their $1$-norms. 
The bounds, in particular, guarantee that the $\rho_m$ are well-defined. 
In addition, they are useful for the proof of Lemma~\ref{lem:hastings}, albeit they are not explicitly needed for the proof of Lemma~\ref{lem:tailored_hastings}. 

\begin{lemma} \label{lem:rho_m_bound}
Let $\rho(G)$ be defined as in Lemma~\ref{lem:rho_in_terms_of_eta} with respect to a Hamiltonian $H$ having a finite interaction (hyper)graph $(V,E)$ with growth constant $\animalc$ and let
\begin{equation} \label{eq:rho_m_in_terms_rhoF}
  \rho_m \coloneqq \sum_{G \in \animals{\geq L}{m}} \rho(G) 
\end{equation} 
for some $F \subset E$.
Then, 
\begin{equation} \label{eq:rho_m_bound}
  \norm{\rho_m}_1 
  \leq \frac{Z(\beta) }{m!} 
  \left(|F| \sum_{l=L}^\infty \alphay(\beta J)^{l} \right)^m ,
\end{equation}
where $\alphay(x) \coloneqq \animalc\, \e^{|x|}(\e^{|x|}-1)$.
\end{lemma}

\begin{proof}
Consider a $k$-fold animal $G \in \animals{\geq L}{k}$ and decompose it into its $k$ non-overlapping animals 
$G_j \in \animals{\geq L}{}$ as 
$G = \dUnion_{j=1}^k G_j \subset E$. 
Then, Eq.~\eqref{eq:rhoF_in_terms_eta} and H\"{o}lder's inequality imply
\begin{equation}
 \norm{\rho(G)}_1 \leq 
 \bnorm{\e^{-\beta \,H_{(\ol G)\complement}}}_1 \prod_{j=1}^k \norm{\eta(G_j)}_\infty ,
\end{equation}
and it follows from Lemma~\ref{lem:gt} and Lemma~\ref{lem:eta_bound} in conjunction with the definition of $\eta$ in Eq.~\eqref{eq:eta_def_new} that
\begin{equation}
 \norm{\rho(G)}_1 \leq Z(\beta)\, y(\beta J)^{|G|} .
\end{equation}
Hence, by the definition from Eq.~\eqref{eq:rho_m_in_terms_rhoF} and Lemma~\ref{lem:mfold_animals_to_animals}, we obtain
\begin{align}
 \norm{\rho_m}_1 &\leq Z(\beta) \sum_{G \in \animals{\geq L}{m}} y(\beta J)^{|G|} \\
 &\leq \frac{Z(\beta)}{m!} \left(\sum_{G \in \animals{\geq L}{}} y(\beta J)^{|G|} \right)^m .
\end{align}
By decomposing the set of animals of size at least $L$ into a union of sets of animals of fixed size $l$, i.e., $\animals{\geq L}{} = \dUnion_{l = L}^\infty \animals{=l}{}$, we can write
\begin{equation}
 \norm{\rho_m}_1 
 \leq \frac{Z(\beta)}{m!} 
    \left(\sum_{l=L}^\infty \left| \animals{=l}{} \right|\, y(\beta J)^{l} \right)^m .
\end{equation}
The bound~\eqref{eq:animal_const} on the number of lattice animals, the fact that the number $|F|$ of edges in $F$ upper bounds the number of possibilities of translating an animal $G$ such that $G \subset F$, and $\alphay = \alpha\, y$ finish the proof.
\end{proof}

%%% ----------------- lem: rho_m --------------------
While the last lemma provides a bound on $\rho_m$ and, in particular, implies that $\rho_m$ is well-defined, the next lemma provides a useful form of $\rho_m$. 

\begin{lemma}\label{lem:rho_m}
  Let $\rho_m$ be defined as in Eq.~\eqref{eq:rho_m_in_terms_rhoF}. Then
  \begin{equation}\label{eq:rho_m}
    \rho_m = 
    \sum_{k=m}^\infty \binom k m \sum_{ w \in \clusters{\geq L}{k} } f(w) \, .
  \end{equation}
\end{lemma}

\begin{proof}
For $G \in \animals{\geq L}{m}$, let 
 \begin{equation}
 \W(G) \coloneqq 
 \{w \in [(\boundary G)\complement]^\ast : G\subset w \} \, .
 \end{equation}
According to Eqs.~\eqref{eq:rhoF_def} and \eqref{eq:rho_m_in_terms_rhoF},
\begin{equation}
 \rho_m = 
 \sum_{G \in \animals{\geq L}{m}}\, \sum_{w \in [(\boundary G)\complement]^\ast: G\subset w} f(w) \, .
\end{equation}
As 
\begin{equation}
 \bigcup_{G \in \animals{\geq L}{m}} \W(G) = \dUnion_{k=m}^\infty \clusters{\geq L}{k} \, , 
\end{equation}
the sums in Eqs.~\eqref{eq:rho_m_in_terms_rhoF} and \eqref{eq:rho_m} contain the same terms.
It remains to show that the multiplicities are correct, i.e., are given by the binomial factor.
Every word in $\W(G)$ contains at least $m$ maximal clusters of size at least $L$, each of which contains a letter in $F$. 
The key is to decompose this set as
\begin{equation}
 \W(G) 
 = \dUnion_{k=m}^\infty  \W^k(G) 
\end{equation}
 with 
\begin{equation*}
 \begin{split}
  \W^k(G) \coloneqq 
    \{ w \in \W(G) : \exists\ &\text{exactly $k$ maximal clusters}\\
	    & c\subset w :  c\in \clusters{\geq L}{ } \} \, ,
  \end{split}
\end{equation*}
i.e., into sets of words having exactly $k\geq m$ such clusters.
Then, the observation that for every $w \in \W^k(G)$ there are exactly $\binom k m$ many $m$-fold animals $G' \in \animals{\geq L}{m}$ with $w \subset G'$ completes the proof.
\end{proof}

%%% ==================================================================
%%% ========================   References  ===========================
%%% ==================================================================

\end{document}